\documentclass[twocolumn,10pt]{article}

\usepackage{times}
\usepackage{datetime}
\usepackage{url}
\usepackage{xcolor}
\usepackage[final]{pdfpages}

\usepackage[ruled,resetcount,noend]{algorithm2e}
\usepackage[show]{chato-notes}
\usepackage{verbatim}
\usepackage{amsmath}
\usepackage{amsthm}
\usepackage{listings}
\usepackage{pgfplots}
\usepackage{pgfplotstable}
\usepackage{pgf-pie}
\usepackage{booktabs}
\usepackage{multirow}

\pgfplotsset{compat=newest}

\definecolor{pastelblue}{RGB}{76,113,175}
\definecolor{pastelgreen}{RGB}{84,167,104}
\definecolor{pastelred}{RGB}{196,78,82}
\definecolor{pastelgrey}{RGB}{230,230,230}
\definecolor{pastelbeige}{RGB}{243,236,221}
\definecolor{pastelpurple}{RGB}{154,139,192}

\pgfplotsset{
  carlos/.style={%
    axis background/.style={fill=white},%
    grid=both,%
    grid style=white
  }%
}

\pgfplotsset{
  every axis/.style=carlos%
}

\newtheorem{theorem}{Theorem}
\newtheorem{definition}{Definition}

\SetKw{Event}{upon}
\SetKwBlock{Do}{}{}
\SetKw{Proc}{procedure}
\SetKw{Before}{before}
\SetKw{After}{after}
\SetKw{Around}{around}
\SetKw{Funct}{function}

\newcommand{\spara}[1]{\noindent\textbf{#1}}

\newcommand{\ezip}{ODAG} 
\newcommand{\ezips}{ODAGs}
\newcommand{\Ezip}{ODAG} 
\newcommand{\Ezips}{ODAGs}
\newcommand{\EZip}{ODAG} 
\newcommand{\EZips}{ODAGs}

\newcommand{\model}{filter-process}

\newcommand{\MODEL}{Filter-Process}

\newcommand{\name}{Arabesque}

\usepackage{pgf}
\usepackage{subfig}

\usepackage{etoolbox}
\newtoggle{singlecolumn}
\newtoggle{twoleveltiny}
\newtoggle{twolevelapprox}
\togglefalse{singlecolumn}
\togglefalse{twoleveltiny}
\togglefalse{twolevelapprox}

\usepackage{titlesec}
\titleformat*{\section}{\large\bfseries}
\titleformat*{\subsection}{\normalsize\bfseries}

\usetikzlibrary{external}

\title{\textbf{\LARGE \name: A System for Distributed Graph Mining\\ Extended version}}

\author{
Carlos H. C. Teixeira\thanks{
Currently a PhD student at the Federal University of Minas Gerais, Brazil. This work was done while the author was at QCRI.}, Alexandre J. Fonseca, Marco
Serafini,\\ Georgos Siganos, Mohammed J. Zaki, Ashraf Aboulnaga\\
\em Qatar Computing Research Institute - HBKU, Qatar}

\date{}

\begin{document}
\includepdf{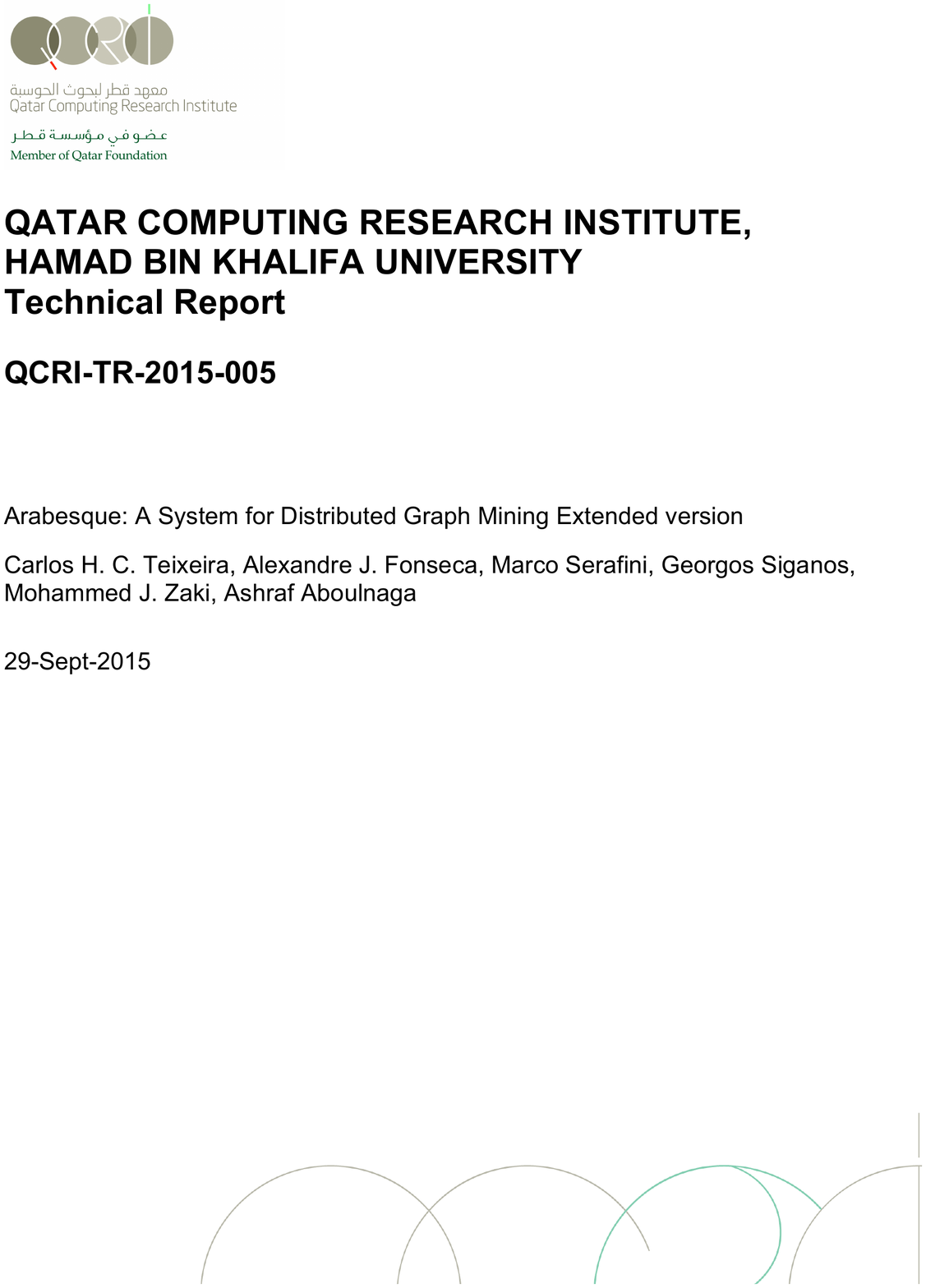}
\includepdf{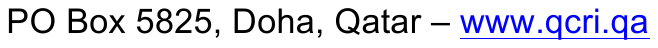}
\clearpage
\setcounter{page}{1}
\maketitle

\begin{abstract}
Distributed data processing platforms such as MapReduce and Pregel have
substantially simplified the design and deployment of certain classes of
distributed graph analytics algorithms. However, these platforms do not
represent a good match for distributed graph mining problems, as for example
finding frequent subgraphs in a graph.
Given an input graph, these problems require exploring a very large number of subgraphs and finding
patterns that match some ``interestingness'' criteria desired by the
user. These algorithms are very important for areas such as social
networks, semantic web, and bioinformatics.

In this paper, we present \name, the first distributed data processing
platform for implementing graph mining algorithms. 
\name\ automates the process of exploring a very large number of subgraphs.
It defines a high-level {\em \model} computational model that simplifies the development of scalable graph mining algorithms: \name\ explores subgraphs and passes them to the application, which must simply compute outputs and decide whether the subgraph should be further extended.
We use \name's API to produce distributed solutions to three fundamental graph mining problems: frequent subgraph mining, counting motifs, and finding cliques. 
Our implementations require a handful of lines of code, scale to trillions of subgraphs, and represent in some cases the first available distributed solutions.
\end{abstract}

\section{Introduction}

Graph data is ubiquitous in many fields, from the Web to advertising and
biology, and the analysis of graphs is becoming increasingly important. The
development of algorithms for graph analytics has spawned a large amount
of research, especially in recent years. However, graph analytics has
traditionally been a challenging problem tackled by expert researchers,
who can either design new specialized algorithms for the problem at
hand, or pick an appropriate and sound solution from a very vast
literature. When the input graph or the intermediate state or computation complexity
becomes very large, scalability is an additional challenge.

The development of graph processing systems such as Pregel~\cite{pregel}
has changed this scenario and made it simpler to design scalable graph
analytics algorithms. Pregel offers a simple ``think like a vertex"
(TLV) programming paradigm, where each vertex of the input graph is a
processing element holding local state and communicating with its
neighbors in the graph. TLV is a perfect match for problems that can be
represented through linear algebra, where the graph is modeled as an
adjacency matrix (or some other variant like the Laplacian matrix) and
the current state of each vertex is represented as a vector. 
We call this class of methods {\em graph computation} problems.
A good example is computing PageRank~\cite{pagerank}, which is based
on iterative sparse matrix and vector multiplication operations. 
TLV covers several other 
algorithms that require a similar computational architecture, for
example, shortest path algorithms, and
over the years many optimizations of this paradigm have been
proposed~\cite{tlv-survey,powergraph,giraph++,blogel}. 

\begin{figure}[!h]
\centering
\begin{tikzpicture}
    \pgfplotstableread[col sep=comma]{num-embeddings-motif-mico.data}\numembeddingsmico
    \pgfplotstableread[col sep=comma]{num-embeddings-motif-youtube.data}\numembeddingsyoutube
    \pgfplotstableread[col sep=comma]{num-embeddings-gspan-citeseer.data}\numembeddingsgspanciteseer
    \pgfplotstableread[col sep=comma]{num-embeddings-cliques-mico.data}\numembeddingscliquesmico
    \pgfplotstableread[col sep=comma]{num-embeddings-motif-sn.data}\numembeddingsmotifsn
    \begin{semilogyaxis}[%
	width=\linewidth,
	height=1.75in,
	legend style={at={(0.5,1.05)}, anchor=south, legend columns=3, font=\scriptsize},
	bar width=1.5em,
	ytick={1000, 10000, 100000, 1000000, 10000000, 100000000, 1000000000, 10000000000, 0100000000000, 1000000000000, 10000000000000},
	yticklabels={$10^3$, , ,$10^6$, , ,$10^9$, , ,$10^{12}$, , ,},
	ylabel=Number of Interesting Subgraphs,
	xlabel=Size of the Subgraphs,
	x label style={font=\scriptsize},
	y label style={font=\scriptsize},
	y tick label style={align=center,font=\scriptsize},
	x tick label style={align=center,font=\scriptsize,text width=4em},
	every node near coord/.append style={fill=white, fill opacity=0.7, anchor=south,font=\tiny, yshift=-0.15em},
	mark options={solid}
	]
    \addplot[color=pastelblue, mark=*]
	table[x=superstep, y=numembeddings]{\numembeddingsmico}; 
    \addplot[color=pastelred, mark=square*]
	table[x=superstep, y=numembeddings]{\numembeddingsyoutube}; 
    \addplot[color=black, mark=triangle*]
	table[x=superstep, y=numembeddings]{\numembeddingscliquesmico}; 
    \addplot[color=pastelgreen, mark=diamond*]
	table[x=superstep, y=numembeddings]{\numembeddingsgspanciteseer}; 
	\addplot[color=pastelgreen, mark=*, dashed]
	table[x=superstep, y=numembeddings]{\numembeddingsmotifsn}; 
    \legend{Motifs (MiCo), Motifs (Youtube), Cliques (MiCo), FSM (CiteSeer), Motifs(SN)}
    \end{semilogyaxis}
\end{tikzpicture}
    \caption{Exponential growth of the intermediate state in graph
        mining problems (motifs counting, clique finding, FSM: Frequent
        subgraph mining) on different datasets.}
    \label{fig:exp-growth}
\end{figure}

Despite this progress, there remains an
important class of algorithms that cannot be readily formulated using
the TLV paradigm. These are {\em graph mining} algorithms used to
discover relevant patterns that comprise both structure-based and
label-based properties of the graph. Graph mining is widely used for
several applications, for example, discovering 3D motifs in protein
structures or chemical compounds, extracting network motifs or
significant subgraphs from protein-protein or gene interaction networks,
mining attributed patterns over semantic data (e.g., in Resource
Description Framework or RDF format), finding structure-content
relationships in social media data, dense subgraph mining for community
and link spam detection in web data, among others. Graph mining algorithms
typically take a labeled and immutable graph as input, and mine
patterns that have some algorithm-specific property (e.g., frequency
above some threshold)
by finding all instances of these
patterns in the input graph. 
Some algorithms also compute aggregated
metrics based on these subgraphs.

Designing graph mining algorithms is a challenging and active area of
research. In particular, scaling graph mining algorithms to even moderately large
graphs is hard. The set of possible patterns and their subgraphs
in a graph can be exponential in the
size of the original graph, resulting in an explosion of the
computation and intermediate state.
Figure~\ref{fig:exp-growth} shows the exponential growth of the number of 
``interesting" subgraphs of different sizes in some of the graph mining problems
and datasets we will evaluate in this paper.
Even graphs with few thousands of edges can quickly generate
hundreds of millions of interesting subgraphs.
The need for enumerating a large number of subgraphs
characterizes graph mining problems and distinguishes them from graph
computation problems.
Despite this state explosion problem, most graph mining
algorithms are centralized because of the complexity of distributed solutions.

In this paper, we propose automatic {\em subgraph exploration} as a generic building
block for solving graph mining problems, and introduce \name, the first
embedding exploration system specifically designed for distributed graph mining.
Conceptually, we move from TLV to ``think like an {\em embedding}''
(TLE), where by embedding we denote a subgraph representing a particular
instance of a more general template subgraph called a {\em pattern} (see Figure~\ref{fig:problem}).

\begin{figure}[t] 
    \centering
    \includegraphics[scale=0.45]{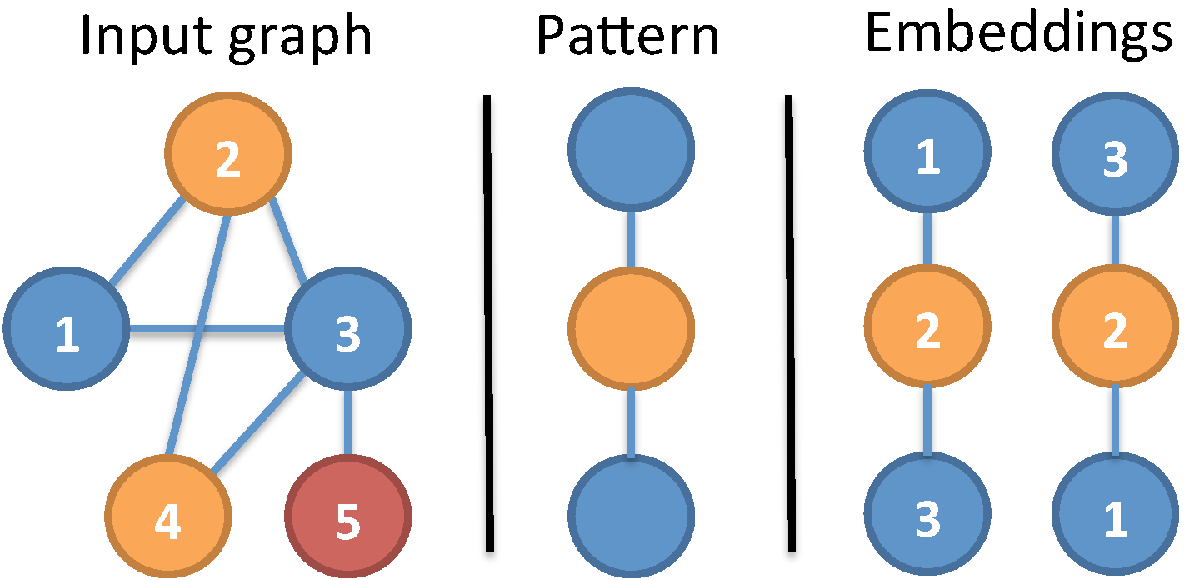}
    \caption{Graph mining concepts: an input graph, an example pattern,
    and the embeddings of the pattern. Colors represent labels. Numbers
denote vertex ids. Patterns and embeddings are two types of subgraphs.
However, a pattern is a {\em template}, whereas an embedding is an {\em
instance}. In this example, the two embeddings are {\em automorphic}.}
    \label{fig:problem}
\end{figure}

\name\ defines a high-level {\em \model} computational model.
Given an input graph, the system takes care of automatically
and systematically visiting all the embeddings that need to be explored
by the user-defined algorithm, performing this exploration in a
distributed manner. 
The system passes all the embeddings it explores to the application, which consists primarily of two functions:
{\em filter}, which indicates whether an embedding should be processed, and
{\em process}, which examines an embedding and may produce some output.
For example, in the case of finding cliques the filter function prunes embeddings that are not cliques, since none of their extensions can be cliques, and the process function outputs all explored embeddings, which are cliques by construction.
\name\ also supports the pruning of the exploration space based on user-defined metrics aggregated across multiple embeddings. 

The \name\ API simplifies and thus democratizes the design of graph mining algorithms, 
and automates their execution in a distributed setting. 
We used \name\ to implement and
evaluate scalable solutions to three fundamental and diverse graph
mining problems: frequent subgraph mining,
counting motifs, and finding cliques.
These problems are defined precisely in Section~\ref{sec:graphmining}.
Some of these algorithms are the first
distributed solutions available in the literature, which shows the
simplicity and generality of \name.

\name's embedding-centered API facilitates a highly scalable implementation. The
system scales by spreading embeddings uniformly across workers, thus
avoiding hotspots. By making it explicit that embeddings are the
fundamental unit of exploration, \name\ is able to use fast
coordination-free techniques, based on the notion of embedding canonicality,
to avoid redundant work and minimize communication costs.
It also enables us
to store embeddings efficiently using a new data structure called
Overapproximating Directed Acyclic Graph (\ezip), and 
to devise a new two-level optimization for pattern-based aggregation, 
which is a common operation in graph mining algorithms.

\name\ is implemented as a layer on top of Apache Giraph~\cite{giraph}, a Pregel-inspired 
graph computation system, 
thus allowing both graph computation and 
graph mining algorithms to run on top of the same infrastructure.
The implementation does not use a TLV approach:
it considers Giraph just as a regular data parallel system
implementing the Bulk Synchronous Processing model.

To summarize, we make the following contributions:
\begin{itemize}
\item
We propose embedding exploration, or ``think like an embedding'', as an effective
basic building block for graph mining.
We introduce the \model\ computational model (Section~\ref{sec:model}),
design an API that enables embedding exploration to be expressed effectively and succinctly,
and present three example graph mining applications that can be
elegantly expressed using the \name\ API (Section~\ref{sec:API}).

\item
We introduce techniques to make distributed embedding exploration scalable:
coordination-free work sharing, efficient  storage of embeddings,
and an important optimization for pattern-based aggregation (Section~\ref{sec:system}).

\item
We demonstrate the scalability of \name\ on various graphs.
We show that \name\ scales to hundreds of cores over a cluster, obtaining
orders of magnitude reduction of running time over the
centralized baselines (Section~\ref{sec:eval}),
and can analyze trillions of embeddings on large graphs.
\end{itemize}

The Arabesque system, together with all applications used for this paper, is publicly available at the project's website: {\small \url{www.arabesque.io}}.

\section{Graph Mining Problems}
\label{sec:graphmining}

In this section, we introduce the graph-theoretic terms we will use
throughout the text and characterize the space of graph mining problems
addressed by \name.

\spara{Terminology:} Graph mining problems take an input graph $G$ where
vertices and edges are labeled. Vertices and edges have unique ids, and 
their labels are arbitrary, domain-specific
attributes that can be null. An {\em embedding} is a subgraph of $G$,
i.e., a graph containing a subset of the vertices and edges of $G$. 
A {\em vertex-induced} embedding is defined starting from a set of vertices
by including all edges of $G$ whose endpoints are in the set.
An {\em edge-induced} embedding is defined starting from a set of edges
by including all the endpoints of the edges in the set. For example,
the two embeddings of Figure~\ref{fig:problem} are induced by the edges
$\{(1,2),(2,3)\}$. In order to be induced by the vertices $\{1,2,3\}$
the embeddings should also include the edge $(1,3)$. We consider only
{\em connected} embeddings such that there is a path connecting each
pair of vertices. 

A {\em pattern} is an arbitrary graph. We say that an
embedding $e$ in $G$ is {\em isomorphic} to a pattern $p$ if and only if there exists a
one-to-one mapping between the vertices of $e$
and $p$, and between the edges of $e$ and $p$, 
such that: (i) each vertex (resp. edge) in $e$ has one matching
vertex (resp. edge) in $p$ with the same labels, and vice versa; (ii)
each matching edge connects matching vertices. In this case, we say
informally that $e$ {\em has pattern} $p$. Equivalently, we say that
there exists a {\em subgraph isomorphism} from $p$ to $G$, and we call
$e$ an {\em instance} of pattern $p$ in $G$.
Two embeddings are {\em
automorphic} if and only if they contain the same edges and vertices,
i.e.,  they are equal and thus also isomorphic (see for example Figure~\ref{fig:problem}).

\spara{Characterizing Graph Mining Problems:}
\name\ targets graph mining problems, which involve subgraph
enumeration.
%
Given an immutable input graph $G$ with labeled
vertices and edges, graph mining problems focus on the enumeration of
all {\em patterns} that satisfy some user-specified ``interestingness''
criteria. A given pattern $p$ is evaluated by listing its matches or
embeddings in the input dataset $G$, and then filtering out the
uninteresting ones. 
Graph mining problems often require
subgraph isomorphism checks to determine the patterns related to sets of
embeddings, and also graph automorphism checks to eliminate duplicate
embeddings.

We consider graph mining problems where one seeks connected graph patterns 
whose embeddings are vertex- or edge-induced. 
There are several variants of these problems.
The input dataset may comprise a collection of many graphs, or a single large graph. The
embeddings of a pattern comprise the set of (exact) isomorphisms from
the pattern $p$ to the input graph $G$. However, in inexact matching,
one can seek inexact or approximate isomorphisms (based on notions of
edit-distances, label costs, etc.). There are many variants in terms of
the interestingness criteria, such as frequent subgraph mining, dense
subgraph mining, or extracting the entire frequency distribution of
subgraphs up to some given number of vertices. Also related to graph
mining is the problem of graph matching, where a query pattern $q$ is
fixed, and one has to retrieve all its matches in the input graph $G$.
Solutions to graph matching typically use indexing approaches, which
pre-compute a set of paths or frequent subgraphs, and use them to
facilitate fast matching and retrieval. As such graph mining encompasses
the matching problem, since we have to both enumerate the patterns and
find their matches. Further, any solution to the single input graph
setting is easily adapted to the multiple graph dataset case. Therefore,
in this paper we focus on graph mining tasks on a single large input
graph. See~\cite{graphminingbook} for a state-of-the-art review of graph
mining and search.

\spara{Use Cases:}
Throughout this paper we consider three classes of problems: {\em
frequent subgraph mining}, {\em counting motifs}, and {\em finding
cliques}. We chose these problems because they represent different
classes of graph mining problems. The first is an example of
explore-and-prune problems, where only embeddings corresponding to a
frequent pattern need to be further explored. Counting motifs
requires exhaustive graph exploration up to some maximum size. Finding
cliques is an example of dense subgraph mining, and allows one to prune 
the embeddings using local information. We discuss these problems below in
more detail. 

Consider the task of {\em frequent subgraph mining (FSM)}, i.e., finding
those subgraphs (or patterns, in our terminology) that occur a minimum
number of times in the input dataset. The occurrences are counted using 
some anti-monotonic function on the set of its embeddings. The
anti-monotonic property states that the frequency of a supergraph should
not exceed the frequency of a subgraph, which allows one to stop
extending patterns/embeddings as soon as they are deemed to be
infrequent. There are several anti-monotonic metrics to measure the
frequency of a pattern. While their differences are not very important
for this discussion, they all require aggregating metrics from all
embeddings that correspond to the same pattern. We use the {\em minimum
image-based support} metric~\cite{imagesupport}, which defines the
frequency of a pattern as the minimum number of distinct mappings for
any vertex in the pattern, over all embeddings of the pattern. Formally, let $G$
be the input graph, $p$ be a pattern, and $\cal E$ its set of
embeddings. The pattern $p$ is frequent if $sup(p,{\cal E}) \ge \theta$,
where $sup()$ is the minimum image-based support function, and $\theta$
is a user-specified threshold. The FSM task is to
mine {\em all} frequent subgraph patterns from a single large graph $G$.

A {\em motif} $p$ is defined as a connected pattern of vertex-induced embeddings
that exists in an input graph $G$. Further, a set of motifs is
required to be non-isomorphic, i.e., there should obviously be no
duplicate patterns. In motif mining~\cite{graphlets}, the input graph
is assumed to be unlabeled, and there is no minimum frequency threshold;
rather the goal is to extract all motifs that occur in $G$ along with
their frequency distribution. Since this task is inherently exponential,
the motifs are typically restricted to patterns of order (i.e.,
number of vertices) at most $k$. For example, for $k=3$ we have two
possible motifs: a chain where ends are not connected and a
(complete) triangle. Whereas the original definition of motifs
targets unlabeled and induced patterns, we can easily generalize the
definition to labeled patterns.

In {\em clique mining} the task is to enumerate all complete subgraphs in the
input graph $G$. A complete subgraph, or clique, $p$ with $k$ vertices is
defined as a connected subgraph where each vertex has degree $k-1$, i.e.,
each vertex is connected to all other vertices (we assume there are no
self-loops).
The clique problem can also be generalized to maximal cliques, 
i.e., those not contained in any other clique, and frequent cliques, if we
impose a minimum frequency threshold in addition to the completeness
constraint.

\section{The \MODEL\ Model}
\label{sec:model}

We now discuss the \model\ model, which formalizes
how \name\ performs embedding exploration based on
application-specific {\em filter} and {\em process} functions.

\begin{algorithm}[t]
\caption{Exploration step in \name.}
\label{alg:exploration}
\begin{footnotesize}
\SetKwInOut{Input}{input}\SetKwInOut{Output}{output}
\Input{Set $I$ of initial embeddings for this step}
\Output{Set $F$ of extended embeddings for the next step (initially empty) }
\Output{Set $O$ of new outputs (initially empty)}
\BlankLine
\ForEach {$e \in I$ such that $\alpha(e)$}{
	add $\beta(e)$ to $O$\;
		$C \leftarrow $ set of all extensions of $e$ obtained by adding one incident edge / neighboring vertex\;
	\ForEach {$e' \in C$}{
		\If{$\phi(e')$ and there exists no $e'' \in F$ automorphic to $e'$}{
		    add $\pi(e')$ to $O$\;
				add $e'$ to $F$\;
		}
	}
}
run aggregation functions\;
\end{footnotesize}
\end{algorithm}

\subsection{Computational Model}
\name\ computations proceed through a sequence of {\em exploration
steps}. At a conceptual level, the system performs the operations
illustrated in Algorithm~\ref{alg:exploration} in each exploration step.
Each step is associated with an initial set $I$ containing
embeddings of the input graph $G$. 
\name\ automates the process of exploring the graph and expanding embeddings.
Applications are specified via two user-defined functions: 
a {\em filter function} $\phi$ and a {\em
process function} $\pi$. 
The application can optionally also define two additional functions 
that will be described shortly.


\name\ starts an exploration by generating the set of {\em candidate} embeddings $C$,
which are obtained by expanding the embeddings in $I$.
The system computes
candidates by adding one incident edge or vertex to $e$, depending on
whether it runs in edge-based or vertex-based exploration mode. In
edge-based exploration, an embedding is an edge-induced subgraph; in
vertex-based exploration, it is vertex-induced (see Section~\ref{sec:graphmining}). 
In the first exploration step, $I$ contains only a special 
``undefined" embedding, whose expansion $C$ consists of all edges or vertices of $G$,
depending on the type of exploration.
The application can
decide between edge-based or vertex-based exploration during
initialization. 

After computing
the candidates, the filter function $\phi$ examines each candidate $e'$ and
returns a Boolean value indicating whether $e'$ needs to be processed. 
If $\phi$ returns true, the process function $\pi$ takes $e'$ as input and outputs a set of user-defined values.
By default, $e'$ is then added to the set $F$. 
After an exploration step is terminated, $I$ is set to be equal to $F$ before the start of the next step.
The computation terminates when the set $F$ is
empty at the end of a step.

\name\ runs the outer loop of Algorithm~\ref{alg:exploration} in
parallel by partitioning the embeddings in $I$ over multiple servers each 
running multiple worker threads.
This distribution is transparent to applications.
Each execution step is executed as a superstep in the Bulk Synchronous Parallel model~\cite{bsp}.

\spara{Optional aggregation functions:}
The \model\ model described so far considers single embeddings in isolation.
A common task in graph mining systems is to aggregate values 
across multiple embeddings, for example grouping embeddings by pattern.
To this end, \name\ offers specific functions to execute user-defined aggregation for multiple embeddings. 
Aggregation can group embeddings by an arbitrary integer value or by pattern,
and is executed on candidate embeddings at the end of the exploration step in which they are generated.

The optional {\em aggregation filter} function $\alpha$ and {\em aggregation process} function
$\beta$ can filter and process an embedding $e$
in the exploration step following its generation. At that time, aggregated information
collected from all the embeddings generated in the same exploration step as $e$ becomes available.
The $\alpha$ function can take filtering decisions before embedding expansion 
based on aggregate values.
For example, in the frequent subgraph mining problem we can use
aggregators to count the embeddings associated with a given
pattern, and then filter out embeddings of infrequent patterns with $\alpha$.
Similarly, the $\beta$ function can be used to output aggregate
information about an embedding, for example its frequency.
By default, $\alpha$ returns {\em true} and $\beta$ does not add any output to $O$.


\spara{Guarantees and requirements:}
Embedding exploration processes every embedding that is not filtered out.
More formally, it guarantees the following {\em completeness}
property: for each embedding $e$ such that $\phi(e) = \alpha(e) = \textit{true}$,
embedding exploration must add $\pi(e)$ and $\beta(e)$ to $O$.

Completeness assumes some properties of the user-defined functions that 
emerge naturally in graph mining algorithms.
The first property is that the application considers two embeddings $e$
and $e'$ to be equivalent if they are automorphic (see
Section~\ref{sec:graphmining}). Formally, \name\ requires what we call
{\em automorphism invariance}: if $e$ and $e'$ are automorphic, then 
each user-defined function must return the same result for $e$ and $e'$.
\name\ leverages this natural property of graph mining algorithms to prune
automorphic embeddings and substantially reduce the exploration space.

The second property \name\ requires is called {\em anti-monotonicity} and is
formally defined as follows: if $\phi(e) = \textit{false}$ then it holds that
$\phi(e') = \textit{false}$ for each embedding $e'$ that extends $e$. 
The same property holds for the optional filter function $\alpha$. 
This is one of the
essential properties for any effective graph mining method
and guarantees that once a filter function
tells the framework to prune an embedding, all its extensions
can be ignored.


\subsection{Alternative Paradigms: Think Like a Vertex and Think Like a Pattern}
\label{sec:tlv-tlp-tle}
We can now contrast our embedding-centric or ``Think Like an Embedding''
(TLE) model with other approaches to build graph mining algorithms.
We empirically compare with these approaches in
Section~\ref{sec:eval}.

The standard paradigm of systems like Pregel is {\em vertex-centric} or
``Think Like a Vertex'' (TLV),
because computation and state are at the level of a vertex in the graph. 
TLV systems are designed to scale for large input graphs:
the information about the input graph is distributed and vertices only have information
about their local neighborhood.
In order to perform embedding exploration, each vertex can keep a set of 
local embeddings, initially containing only the vertex itself.
Then, to expand a local embedding $e$, vertices push it to the ``border" vertices of $e$
that know how to expand $e$ by adding its neighbors.
Each expansion results in a new embedding, 
which is sent again to border vertices and so on.
With this approach, highly connected vertices must take on a disproportionate fraction of embeddings to expand.
The approach also creates a significant number of duplicate messages
because each new embedding must be sent to all its border vertices. 
These limitations significantly affect performance. 
In our experiments, we observed that TLV-based embedding exploration algorithms can be
two orders of magnitude slower compared to TLE.

The current state-of-the-art centralized methods for solving graph mining tasks typically adopt a
different, {\em pattern-centric} or ``Think Like a Pattern'' (TLP) approach. 
The key difference between TLP and the embedding-centric view 
of the \model\ model is that it is not necessary to 
explicitly materialize all embeddings: state can be kept at the granularity
of patterns (which are much fewer than embeddings) and embeddings may
be re-generated on demand.
The process starts with the set of
all possible (labeled) single vertices or edges as candidate patterns. 
It then processes embeddings of each pattern, often by recomputing them on the fly.
After aggregation and pattern filtering, the valid set of
patterns are extended by adding one more vertex or edge. Subgraph or pattern
mining proceeds iteratively via recursive extension, processing and
filtering steps, and continues until no new patterns are found.
Parallelizing the computation via partitioning it by pattern can easily
result in load imbalance, as our experiments show. This is because there
 are often only few patterns that are highly popular -- indeed, finding these
few patterns is the very goal of graph mining.
These popular patterns result in hotspots among workers and thus in poor load balancing.


\section{\name: API, Programming, and Implementation} \label{sec:API}
We now describe the \name\ Java API and show how we use it to implement
our example applications.

\subsection{Arabesque API}
The API of \name\ is illustrated in Figure~\ref{fig:API}. 
The user must implement two functions: \texttt{filter},
which corresponds to $\phi$, and \texttt{process}, which
corresponds to $\pi$. 
The process function in the API is responsible for adding results to the output by invoking the \texttt{output} function provided by \name, which prints the results to the underlying file system (e.g. HDFS).
The optional functions $\alpha$ and $\beta$ correspond, respectively, to \texttt{aggregationFilter} and \texttt{aggregationProcess}.
These application-specific functions
are invoked by the \name\ framework as illustrated in
Algorithm~\ref{alg:exploration}.
All these functions have access to a local read-only copy of the graph.


\lstset{language=Oz} 
\lstset{basicstyle=\ttfamily\scriptsize,breaklines=true,
        keywordstyle=\color{blue},otherkeywords={process,reduce,termination,reduceOutput,filter,
        post,aggregation,output,map,readAggregate,mapOutput,aggregationFilter,aggregationProcess,terminationFilter,
        getAutomorphisms}}

\begin{figure}[!h]
\begin{footnotesize}
\underline{Mandatory application-defined functions:}
\begin{lstlisting}
boolean filter (Embedding e)
void process (Embedding e)
\end{lstlisting}
\underline{Optional application-defined functions:}
\begin{lstlisting}
boolean aggregationFilter (Embedding e)
void aggregationProcess (Embedding e)
Pair<K,V> reduce (K key, List<V> values)
Pair<K,V> reduceOutput (K key, List<V> value)
boolean terminationFilter (Embedding e)
\end{lstlisting}
\underline{\name\ functions invoked by applications:}
\begin{lstlisting}
void output (Object value)
void map (K key, V value)
V readAggregate (K key)
void mapOutput (K key, V value)
\end{lstlisting}
\end{footnotesize}
\caption{Basic user-defined functions in \name.}
\label{fig:API}
\end{figure}

For performance and scalability reasons, a MapReduce-like model is used to compute aggregated values. 
Applications can send data to reducers using the \texttt{map} function, which is part of 
the \name\ framework and adds a value to an aggregation group defined by a certain key. 
Many applications use the pattern of an embedding as the aggregation key.
\name\ detects when the key is a pattern and uses specific optimizations to make this aggregation efficient, as explained in Section~\ref{sec:optim}.
The application specifies the aggregation logic through the \texttt{reduce} function.
This function receives all values mapped to a specific key in one execution step and aggregates them.
Any method can read the values aggregated over the previous execution step using the \texttt{readAggregate} method.

Output aggregation is a special case where aggregated values are sent directly to the underlying distributed filesystem at the end of each exploration step and are not made available for later reads. 
It can be used through the methods \texttt{mapOutput} and \texttt{reduceOutput}.
Their logic is similar to the aggregation functions described previously, 
but aggregation is only executed when the whole computation ends.

The \texttt{terminationFilter} function can halt the computation of an embedding
when some pre-defined condition is reached.
\name\ applies this filter on an embedding $e$ after executing $\pi(e)$ and before adding $e$ to $F$. 
This is just an optimization to avoid unnecessary exploration steps.
For example, if we are interested in embeddings of maximum size $n$,
the termination filter can halt the computation after processing embeddings of size $n$ at step $n$. 
Without this filter, the system would have to proceed to step $n+1$ and generate all embeddings of size $n+1$
just to filter all of them out.



\subsection{Programming with Arabesque}
\label{sec:applications}

We used the \name\ API to implement
algorithms that solve the three problems discussed in
Section~\ref{sec:graphmining}. The pseudocode of the implementations is
in Figure~\ref{fig:applications}. 
Each of the applications consists of very few lines of code, a stark contrast
compared to the complexity of the specialized state of the art
algorithms solving the same problems~\cite{gspan,graphlets,bron-kerbosch}.

\lstset{language=Oz} 
\lstset{basicstyle=\ttfamily\scriptsize,breaklines=true,
        keywordstyle=\color{blue},otherkeywords={process,reduce,reduceOutput,filter, 
        post, aggregationFilter, aggregationProcess}}
\begin{figure}[!h]
\begin{footnotesize}
    \begin{lstlisting} 
    boolean filter(Embedding e) { return true; }
    void process(Embedding e){ 
        map (pattern(e), domains(e)); } 
    Pair<Pattern,Domain> reduce
                 (Pattern p, List<Domain> domains){
        Domain merged_domain = merge(domains);
        return Pair (p, merged_domain); } 
    boolean aggregationFilter(Embedding e){ 
    	Domain m_domain = readAggregate(pattern(e));
    	return (support(m_domain)>=THRESHOLD); }
    void aggregationProcess(Embedding e) {
        output(e); }
    \end{lstlisting} 
    \begin{center} 
    (a) \textbf{Frequent subgraph mining} (edge-based exploration) 
    \end{center} 
    \vspace{-2mm}
    \begin{lstlisting} 
    boolean filter(Embedding e){ 
        return (numVertices(e) <= MAX_SIZE);} 
    void process(Embedding e){ 
    	mapOutput (pattern(e),1); } 
    Pair<Pattern,Integer> reduceOutput
           (Pattern p, List<Integer> counts){ 
    	return Pair (p, sum(counts)); }
    \end{lstlisting} 
    \begin{center} 
    \vspace{-4mm} 
    (b) \textbf{Counting motifs} (vertex-based exploration)
    \end{center} 
    \vspace{-2mm} 
    \begin{lstlisting} 
    boolean filter (Embedding e){
    	return isClique(e); } 
    void process (Embedding e){ output(e); } 
    \end{lstlisting} 
    \begin{center}
    \vspace{-4mm} 
    (c) \textbf{Finding cliques} (vertex-based exploration) 
    \end{center} 
\end{footnotesize}
\caption{Examples of \name\ applications.}
\label{fig:applications} 
\end{figure}

In frequent subgraph mining we use aggregation to calculate the support function described in Section~\ref{sec:graphmining}.
The support metric is based on the notion of domain, 
which is defined as the set of distinct mappings between
a vertex in $p$ and the matching vertices in any automorphism of $e$.
The \texttt{process} function invokes \texttt{map} to send the domains of $e$ 
to the reducer responsible for the pattern $p$ of $e$. 
For example, in Figure~\ref{fig:problem} the domain of the blue vertex on the top
of the pattern is vertex 1 in the first embedding and 3 in the second.
The function \texttt{reduce} merges all domains: 
the merged domain of a vertex in $p$ is the union of all its aggregated mappings.
Since expansion is done by adding one edge in each exploration step, we
are sure that all embeddings for $p$ are visited and processed in
the same exploration step. The \texttt{aggregationFilter} function reads the
merged domains of $p$ using \texttt{readAggregate} and computes the support, which
is the minimum size of the domain of any vertex in $p$.
It then filters out embeddings for patterns that do not have enough support.
Finally, the \texttt{aggregationProcess} function outputs all the embeddings having a frequent pattern (those that survive the aggregation-filter).
The implementation of this application consists of 280 lines of Java code. 
Of these, 212 are related to handling domains and computing support,
which are basic tasks required in any algorithm for frequent subgraph mining 
to characterize whether an embedding is relevant.
By comparison, the centralized baseline we use for evaluation, 
GRAMI~\cite{grami}, consists of 5,443 lines of Java code.

For motif frequency computation, we perform an exhaustive exploration
of all embeddings until we reach a given maximum size and count all
embeddings having the same pattern. Since the input graph is  not labeled
in this case, a pattern corresponds to a motif. 
The function \texttt{mapOutput} sends a value to the output reducer 
responsible for the pattern of $e$.
The \texttt{reduceOutput} function  outputs the sum of the counts for each motif $p$.
Our implementation consists of 18 lines of code, very few compared to the 
3,145 lines of C code of our centralized baseline (Gtries~\cite{gtries}).

In finding cliques we do a local pruning: if an embedding is not a
clique, none of its extensions can be a clique. Since we visit only
cliques, the evaluation function outputs the embeddings it
receives as input.
The \texttt{isClique} function checks that the newly added vertex 
is connected with all previous vertices in the embedding.
This application consists of 19 lines of code, while
our centralized baseline (Mace~\cite{mace}) consists of 4,621 lines of C code.


In all these examples it is easy to verify that the evaluation and
filter functions satisfy the anti-monotonic and automorphism invariance
properties required by the \name\ computational model.



\subsection{\name\ implementation}
\label{sec:implementation}

\name\ can execute on top of any system supporting the BSP model.
We have implemented \name\ as a layer on top of Giraph.
The implementation does not follow the TLV paradigm: 
we use Giraph vertices simply as {\em workers} that bear no relationship to any
specific vertex in the input graph. 
Each worker has access to a copy of the whole input graph 
whose vertices and edges consist of incremental numeric ids. 
Communication among workers occurs
in an arbitrary point-to-point fashion and communication edges are not
related to edges in the input graph.
Giraph computations proceed through synchronous supersteps according to the BSP model: 
at each superstep, workers first receive all messages sent in the previous
superstep, then process them, and finally send new messages to be
delivered at the next superstep.
The operations of the workers are described in
Section~\ref{sec:system}.
Aggregation functions, if specified, are executed using standard Giraph aggregators.
Optional output workers are used for applications that aggregate output values.
The input values for output aggregation persist over supersteps. Once
the computation is over, output workers aggregate all their input values
and output them.

\section{Graph Exploration Techniques} \label{sec:system}


We now describe in more detail how \name\ performs graph exploration.
We first discuss the coordination-free exploration strategy used by workers 
to avoid redundant work.
We then introduce the techniques we use to store and partition embeddings efficiently.

\subsection{Coordination-Free Exploration Strategy} 
When running exploration in a distributed setting, multiple workers can
reach two ``identical", i.e., automorphic (see Section~\ref{sec:graphmining}), 
embeddings through different exploration paths. Consider
for example the graph of Figure~\ref{fig:problem}. Two different workers $w_1$ and
$w_2$ may reach the two embeddings in Figure~\ref{fig:problem}, one
starting from edge $(1,2)$ by adding $(2,3)$ and the other starting from
edge $(3,2)$ by adding $(2,1)$.
Since all user-defined functions are automorphism-invariant 
(see Section~\ref{sec:model}) we can avoid redundant work 
by discarding all but one of the identical, automorphic 
embeddings.

\name\ solves this problem using a novel coordination-free scheme based on the notion of 
{\em embedding canonicality}.
Informally, we need to select exactly one
of the redundant automorphic embeddings and elect it as ``canonical".
In our example, before $w_1$ and $w_2$ execute the filter and process functions
on a new embedding $e$, \name\ 
executes an embedding canonicality check to verify whether $e$
can be pruned (see
Algorithm~\ref{alg:exploration}). 
This check runs on a single embedding without requiring coordination, as we now discuss.

A sound canonicality check must have the property of {\em uniqueness}: given the
set $S_e$ of all embeddings automorphic to an embedding $e$, there is
exactly one canonical embedding $e_c$ in $S_e$. We call $e_c$ the {\em
canonical automorphism} of $e$.
In \name\ we also need an additional property for canonicality checks
called {\em extendibility}, which we define as follows. Let $e$ be a
candidate embedding obtained by extending a parent embedding $e'$ by one
vertex or edge. The parent embedding $e'$ is canonical because it has
not been pruned. Let $e_c$ be the canonical automorphism of $e$.
Extendibility requires that $e_c$ is one of the extensions of $e'$. This
allows \name\ to prune the automorphisms of a parent $e'$
while still exploring the canonical automorphism of each child $e$.


%
%

\begin{algorithm}[!ht] \caption{\name's incremental embedding canonicality
    check (vertex-based exploration)} \label{alg:check-canonical}
    \begin{footnotesize}
        \SetKwInOut{Input}{input}\SetKwInOut{Output}{output}
        \Input{Input graph $G$} \Input{Canonical parent embedding $\langle
        v_1,\ldots,v_n \rangle$} \Input{Extension vertex $v$}
        \Output{$true$ iff $\langle v_1,\ldots,v_n, v \rangle$ is
        canonical}

\BlankLine

\If {$v_1 > v$ } { \Return $\textit{false}$\; }
$\textit{foundNeighbour} \leftarrow \textit{false}$\; 
\For{$i = 1 \ldots n$} {
	\If{ $\textit{foundNeighbour} = \textit{false}$ and $v_i$ neighbor of $v$ in $G$}{ 
		$\textit{foundNeighbour} \leftarrow \textit{true}$\;  } 
	\ElseIf{$\textit{foundNeighbour} = \textit{true}$ and $v_i > v$}{
		\Return $\textit{false}$\; } 
}

\Return $true$\; \end{footnotesize}
\end{algorithm}

\name\ checks embedding canonicality for each candidate
before applying the filter function.
There can be a huge number
of candidates, so it is essential that the check
is efficient.
We developed a linear-time algorithm, which is based on the 
following definition of canonical embedding.

Consider the case of vertex-based exploration (the edge-based case is analogous).
An embedding $e$ is canonical if and only if its vertices were visited in the following order:
start by visiting the vertex with the smallest id, and then recursively add
the neighbor in $e$ with smallest id that has not been visited yet.
For better performance, \name\ performs this canonicality check in an incremental fashion,
as illustrated in Algorithm~\ref{alg:check-canonical}.
When a worker processes an embedding $e \in I$, it can already assume that
$e$ is canonical because \name\ prunes non-canonical embeddings before 
passing them on to the next exploration step. \name, characterizes an embedding 
as the list of its vertices sorted by the order in which they
have been visited -- the embedding is vertex-induced so the list uniquely
identifies it.
When \name\ checks the canonicality of a new candidate embedding 
obtained by adding a vertex $v$ 
to a parent canonical embedding $e$,
the algorithm scans $e$ in search for the first neighbor $v'$ of $v$,
and then vertifies that there is no vertex in $e$ after $v'$ with higher id than $v$.

The Appendix section of this paper includes proofs showing
that Algorithm~\ref{alg:check-canonical} satisfies the uniqueness and extendibility
properties of canonicality checking.
We also show that these two properties, together with anti-monotoni\-city and 
automorphism invariance, are sufficient
to ensure that \name\ satisfies the completeness property of embedding exploration
(see Section~\ref{sec:model}).

\subsection{Storing Embeddings Compactly}
\label{subsec:lossy-trees}

Graph mining algorithms can easily generate trillions of embeddings as intermediate state.
Centralized algorithms typically do not explicitly store the embeddings they have explored.
Instead, they use application-specific knowledge to rebuild embeddings on the fly 
from a much smaller state.

The only application-level logic available to \name\ consists of the filter and process functions,
which are opaque to the system.
After each exploration step, \name\ receives a set of embeddings filtered by the application
and it needs to keep them in order to expand them in the next step.
Storing each embedding separately can be very expensive.
As we have seen also in Algorithm~\ref{alg:check-canonical}, \name\ represents embeddings as sequences of numbers, representing vertex or edge ids depending on whether exploration is vertex-based or edge-based.
Therefore, we need to find techniques to store sets of sequences of integers efficiently.

Existing compact data structures such as prefix-trees are too expensive because
we would still have to store a new leaf for each embedding 
in the best case, since all canonical embeddings we want to represent
are different.
In contrast, \name\ uses a novel technique called Overapproximating Directed Acyclic Graphs, or \ezips.
At a high level, \ezips\ are similar to prefix trees where all nodes at the same depth corresponding to the same vertex in the graph are collapsed into a single node.
This more compact representation is an overapproximation (superset) of the set of sequences we want to store.
When extracting embeddings from \ezips\ we must do extra work to discard spurious paths.
\Ezips\ thus trade space complexity for computational complexity.
This makes \ezips\ similar to the representative sets introduced in~\cite{repsets}, which
have a higher compression capability but require more work to filter out spurious embeddings. 
In addition, it is harder to achieve effective load balancing with representative sets.
We now discuss the details of \ezips\ and show how they are used in \name.

\spara{The \EZip\ Data Structure:}
For simplicity, we focus the discussion on \ezips\ for vertex-based exploration; the edge-based case is analogous. 
The \ezip\ for a set of canonical embeddings consists of as many arrays as the number of vertices of all the embeddings.
The $i^{th}$ array contains the ids of all vertices in the $i^{th}$ position in any embedding.
Vertex $v$ in the $i^{th}$ array is connected to vertex $u$ in the $(i+1)^{th}$ array if there is at least one canonical embedding with $v$ and $u$ in position $i$ and $i+1$ respectively in the original set.
An example of \ezip\ is shown in Figures~\ref{fig:ezip-graph} and~\ref{fig:ezip}.

\begin{figure}[!h]
\centering
\includegraphics[scale=0.4]{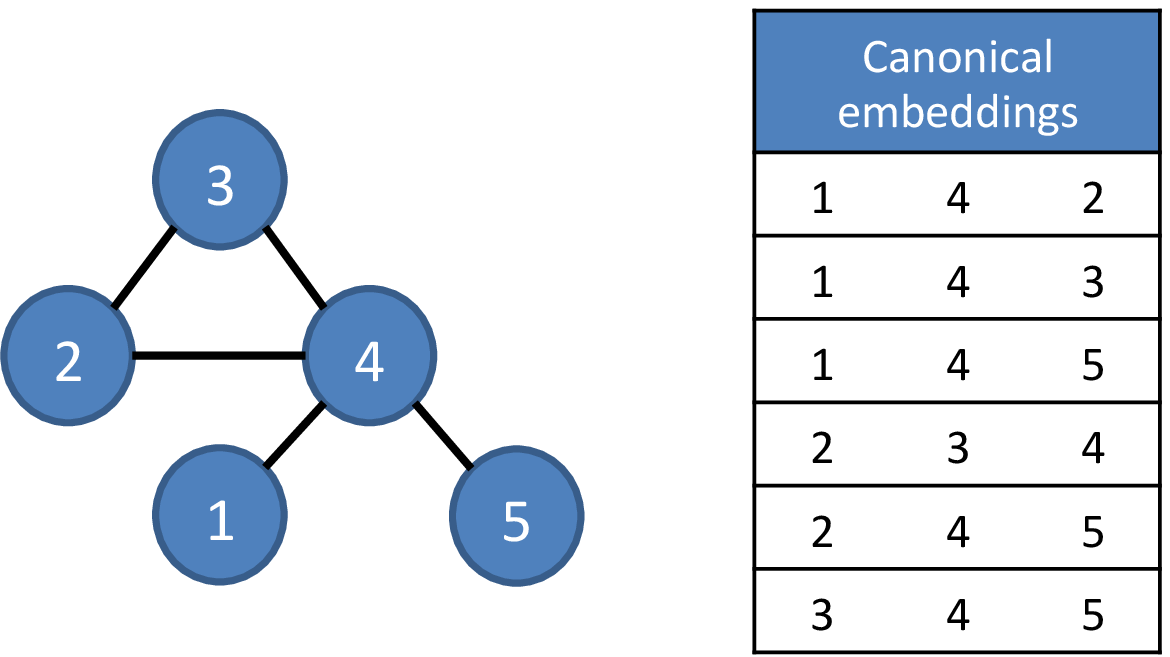}  
\caption{Example graph and its set of $S$ of canonical vertex-induced embeddings of size 3.}
\label{fig:ezip-graph}
%
\includegraphics[scale=0.4]{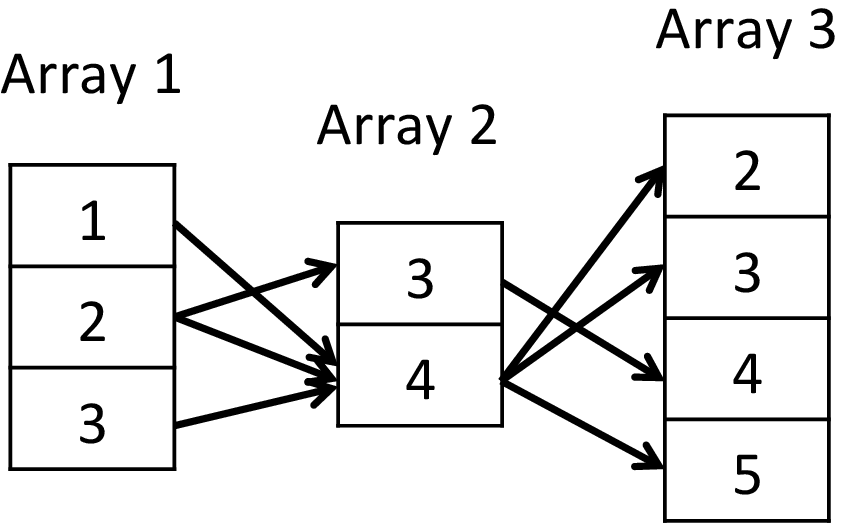}  
\caption{\Ezip\ for the example of Figure~\ref{fig:ezip-graph}. It also encodes spurious embeddings such as $\langle 3,4,2 \rangle$.}
\label{fig:ezip}
\end{figure}

Storing an \ezip\ is more compact than storing the full set of embeddings.
In general, in a graph with $N$ distinct vertices we can have up to $O(N^k)$ different embeddings of size $k$.
With \ezips\ we only have to keep edges between $k$ arrays, where $k$ is the size of the embeddings, so the upper bound on the size is $O(k \cdot N^2) = O(N^2)$ if $k$ is a constant much smaller than $N$.

It is possible to obtain all embeddings of the original set by simply following the edges in the \ezip.
However, this will also generate spurious embeddings that are not in the original set.
Consider for example the graph of Figure~\ref{fig:ezip-graph} and its set of canonical embeddings $S$.
Expanding the \ezip\ of Figure~\ref{fig:ezip} generates also $\langle 3,4,2 \rangle \not\in S$.
Filtering out such spurious embeddings requires application-specific logic.

\spara{\EZips\ in \name:}
In \name, workers produce new embeddings in each exploration step and add them to the set $F$ (see Algorithm~\ref{alg:exploration}).
Workers use \ezips\ to store $F$ at the end of an exploration step, and extract embeddings from \ezips\
at the beginning of the next step.

Filtering out spurious embeddings, as discussed, requires application-specific logic.
Applications written using the filter-process model give \name\ enough information to perform filtering.
In fact, workers can just apply the same filtering criteria as Algorithm~\ref{alg:exploration}: 
the canonicality check and the user-defined filter and aggregate filter functions.
If any of these checks is negative for an embedding, we know that the embedding itself 
or, thanks to the anti-monotonicity property, one of its parents, was filtered out.
In addition, since our canonicality check is incremental, we do not need to consider embeddings 
that extend a non-canonical sequence, so we can prune multiple embeddings at once.

%

After every exploration step,
\name\ merges and broadcasts the new embeddings, and thus the \ezips, to every worker.
In order to reduce the number of spurious embeddings, 
workers group their embeddings in one \ezip\ per pattern.
For each pattern, workers merge their local per-pattern \ezips\ into a single per-pattern global \ezip.
Each per-pattern global \ezip\ is replicated at each worker before the beginning of the next exploration step.

Merging \ezips\ requires merging edges obtained by different workers.
For example, consider the first two arrays in Figure~\ref{fig:ezip}.
One worker might have explored the edge $\langle 2, 3 \rangle$ and another worker $\langle 2, 4 \rangle$.
Edges of \ezip\ arrays are indexed by their initial vertex, so the two arrays from the two workers will have two different entries for the element $2$ of the first array that need to be merged.
Edge merging is very expensive because of the high
number of edges in an \ezip.
Therefore, \name\ uses a map-reduce step to execute edge merging.
Each entry of each array produced by a worker is mapped to the worker associated to its index value.
This worker is responsible for merging all edges for that entry produced by all workers into a single entry.
The entry is then broadcast to all other workers, which computes the union of all entries in parallel.

\subsection{Partitioning Embeddings for Load Balancing}
After broadcast and before the beginning of the next exploration step, every worker obtains the same set of \ezips, one for each pattern mined in the previous execution step. 
The next step is to partition the set $I$ of new embeddings (see Algorithm~\ref{alg:exploration}) 
among workers.
This is achieved by partitioning the embeddings in each pattern \ezip\ separately.

Workers could achieve perfect load balancing by using a round-robin strategy to share work: worker 1 takes the first embedding, worker 2 the second and so on.
However, having workers iterate through all embeddings produced in the previous step, including those that they are not going to process, is computationally intensive.
Therefore, workers do round robin on large blocks of $b$ embeddings. 
The question now is how to identify such blocks efficiently.

\name\ associates each element $v$ in every array with an estimate of how many embeddings, canonical or not, can be generated starting from $v$.
To this end, \name\ assigns a cost 1 to every element of the last array, and it assigns to an element of the $i^{th}$ array the sum of the costs of all elements in the $(i+1)^{th}$ array it is connected to.
Load can then be balanced by having each worker take a partition of the elements in the first array that has approximately the same estimated cost.
While partitioning, it could happen that the cost of an element $v$ of the first array needs to be split among multiple workers.
In this case, the costs associated to the elements of the second array connected to $v$ are partitioned. 
The process continues recursively on subsequent arrays until a balanced load is reached.

\subsection{Two-Level Pattern Aggregation for Fast Pattern Canonicality Checking} 
\label{sec:optim}

\name\ uses a special optimization to speed up per-pattern aggregation, as discussed in Section~\ref{sec:API}.
The optimization was motivated by the high potential cost of this type of aggregation, as we now discuss.

Consider again the example of Figure~\ref{fig:problem} and assume that we want to count the instances
of all single-edge patterns.
The three single-edge embeddings $(1,2)$, $(2,3)$, and $(3,4)$
should be aggregated together since they all have a blue and a yellow endpoint.
Therefore, their two patterns $(blue,yellow)$ and $(yellow,blue)$ should be considered equivalent because they are isomorphic (see Section~\ref{sec:graphmining}).
The aggregation reducer for these two patterns is associated to a single {\em canonical pattern} that is isomorphic to both.
Mapping a pattern to its canonical pattern thus entails solving the
graph isomorphism problem, for which no polynomial-time solution is known~\cite{garey2002computers}.
This makes pattern canonicality much harder than embedding canonicality, 
which is related to the simpler graph automorphism problem.

Identifying a canonical pattern for each single candidate embedding would be a significant bottleneck for \name, as we show in our evaluation, because of the sheer number of candidate embeddings that are generated at each exploration step.
\name\ solves this problem by using a novel technique called {\em two-level pattern aggregation}. 

The first level of aggregation occurs based on what we call {\em quick} patterns.
A quick pattern of an embedding $e$ is the one obtained, in linear time, by simply scanning all vertices (or edges, depending on the exploration mode) of $e$ and extracting the corresponding labels.
The quick pattern is calculated for each candidate embedding.
In the previous example, we would obtain the quick pattern $(blue,yellow)$ for the embeddings $(1,2)$ and $(3,4)$ and the quick pattern $(yellow,blue)$ for the embedding $(2,3)$.
Each worker locally executes the reduce function based on quick patterns.
Once this local aggregation completes, a worker computes the canonical pattern $p_c$ 
for each quick pattern $p$ and sends the locally aggregated value to the reducer for $p_c$. 
\name\ uses the {\em bliss} library to determine canonical patterns~\cite{bliss}.

In summary, instead of executing graph isomorphism for a very large number of candidate embeddings, two-level pattern aggregation computes a quick pattern for every embedding, obtains a number of quick patterns which is orders of magnitude smaller than the candidate embeddings, and then calculates graph isomorphism only for quick patterns.

\section{Evaluation} \label{sec:eval}

\subsection{Experimental Setup}

\spara{Platform:}
We evaluate Arabesque using a cluster of $20$ servers. Each server has
2 Intel Xeon E5-2670 CPUs with a total of 32 execution threads at
2.67GHz per core and 256GB RAM. The servers are connected with a 10 GbE
network.
Hadoop 2.6.0 was configured so that each physical 
server contains a single worker which can use all 32 execution threads (unless otherwise stated). 
\name\ runs on Giraph development trunk from January 2015
with added functionality for obtaining cluster deployment details and improving
aggregation performance. These modifications amount to 10 extra lines of code.


\begin{table}[!ht]
\centering
\begin{small}
    \begin{tabular}{@{}lrrrr@{}} 
    \toprule 
    	 		& \textbf{Vertices} 		& \textbf{Edges} 		& \textbf{Labels} 	& \textbf{Av. Deg.} \\ 
    	 		\toprule 
    	CiteSeer 	& 3,312 		& 4,732 		& 6 		& 2.8 \\ 
    	MiCo  		& 100,000 		& 1,080,298 	& 29 		& 21.6 \\
		Patents	& 2,745,761		& 13,965,409	& 37 & 10 \\
		Youtube	& 4,589,876		& 43,968,798	& 80 & 19 \\
		SN 		    & 5,022,893 & 198,613,776   & 0  & 79\\
		Instagram & 179,527,876 & 887,390,802 & 0 & 9.8 \\

       \bottomrule
    \end{tabular}
\end{small}
\caption{Graphs used for the evaluation.} \label{tab:data-details} \end{table}

\spara{Datasets:}
We use six datasets (see Table \ref{tab:data-details}). 
CiteSeer~\cite{grami} has publications as vertices, with their Computer Science area
as label, and citations as edges. 
MiCo~\cite{grami} has authors as vertices, which
are labeled with their field of interest, and co-authorship of a paper
as edges. 
Patents~\cite{patent_data} contains citation edges between US Patents between January 1963 and
December 1999; the year the patent was granted is considered to be the label.
Youtube~\cite{youtube_data} lists crawled video ids and related videos for each video posted from February 2007 to July 2008. The label is a combination of the video's rating and length.
SN, is a snapshot of a real world Social Network, which is not publicly available.
Instagram is a snapshot of the popular photo and video sharing social network 
collected by~\cite{hamed_2015}. 
 We consider all the graphs to be undirected. 
Note that even if some of these graphs are not 
very large, the explosion of the intermediate computation and state
required for graph exploration (see Figure~\ref{fig:exp-growth}) makes
them very challenging for
centralized algorithms.

\spara{Applications and Parameters:}
We consider the three applications discussed in Sections~\ref{sec:graphmining}, which we label FSM, Motifs and Cliques.
By default, all Motifs executions are run with a maximum embedding size
of 4, denoted as MS=4, whereas
Cliques are run with a maximum embedding size of MS=5. For FSM, 
we explicitly state the support, denoted S, used in each experiment
as this parameter is very sensitive to the properties of the input graph.

\subsection{Alternative Paradigms: TLV and TLP} 
 We start by motivating the necessity for a new framework for distributed 
 graph mining.
We evaluate the two alternative computational paradigms that we 
discussed in Section~\ref{sec:tlv-tlp-tle}. \name\ (i.e., TLE) will be evaluated in the next subsection. 
We consider the problem of 
frequent subgraph mining (FSM) as a use case.
Note that there are currently no distributed solutions to solve FSM on a single large input graph in the literature. 



\begin{figure}[!ht]
\centering
\begin{tikzpicture}
    \pgfplotstableread[col sep=comma]{scalability-tle.data}\scalabilitytle
    \pgfplotstableread[col sep=comma]{scalability-tlp.data}\scalabilitytlp
    \pgfplotstableread[col sep=comma]{scalability-tlv.data}\scalabilitytlv
    \pgfplotstableread[col sep=comma]{
	nodes, speedup
	1, 1
	5, 5
	10, 10
	}\speedupideal
    \begin{axis}[%
	width=\linewidth,
	height=1.65in,
	legend style={at={(0.33,0.75)}, anchor=south, legend columns=3, font=\scriptsize},
	bar width=1.5em,
	xmin=1,
	ymin=0,
	ymax=10,
	xmax=10,
	ylabel=Speedup,
	xlabel=Number of nodes (32 threads),
	x label style={font=\scriptsize},
	y label style={font=\scriptsize},
	y tick label style={align=center,font=\scriptsize},
	x tick label style={align=center,font=\scriptsize,text width=4em},
	every node near coord/.append style={fill=white, fill opacity=0.7, anchor=south,font=\tiny, yshift=-0.15em},
	enlarge x limits=0.075,
	xtick=data,
	mark options={solid}
	]
    \addplot[color=black, mark=*]
	table[x=nodes, y=speedup]{\scalabilitytle}; 
    \addplot[color=pastelblue, mark=diamond*]
	table[x=nodes, y=speedup]{\scalabilitytlp}; 
    \addplot[color=pastelred, mark=triangle*]
	table[x=nodes, y=speedup]{\scalabilitytlv}; 
    \legend{Ideal, TLP, TLV}
    \end{axis}
\end{tikzpicture}
    \caption{Scalability Analysis of Alternative Paradigms: FSM (S=300) on CiteSeer.}%
    \label{fig:scalability_tlv}
\end{figure}


\textbf{The Case of TLV:} Our TLV implementation 
globally maintains the set of embeddings that have been visited, much like \name.
The implementation adopts the TLV approach as described in Section~\ref{sec:tlv-tlp-tle}
and uses the same coordination-free technique as \name\ to avoid redundant work.
The TLV implementation also uses application-specific approaches 
to control the expansion process. Our TLV implementation of FSM uses this feature to follow 
the standard depth-first strategy of gSpan~\cite{gspan}. 

In Figure~\ref{fig:scalability_tlv}, we show the scalability of FSM with support 300 using
the CiteSeer graph. 
As seen from the figure, TLV does not scale beyond 5 servers. A major scalability bottleneck 
is that each embedding needs to be replicated to each vertex that has the necessary local information 
to expand the embedding further. In addition, high-degree vertices need to expand a 
disproportionate fraction of embeddings. CiteSeer is a scale-free graph thus affecting the scalability 
of TLV. 

Overall TLV performance is two
orders of magnitude slower compared to \name. TLV requires more
than $300$ seconds to run FSM on the CiteSeer graph, while \name\
requires only $7$ seconds for the same setup. 
The total messages
exchanged for this tiny graph is $120$ million, versus $137$ thousand
messages required by \name. Due to the hot-spots inherent to the graph
structure, or the label distribution, and the extended duplication of
state that the TLV paradigm requires, we conclude that TLV is not suited for
solving these problems.




\textbf{The Case of TLP:} The TLP implementation is based on GRAMI~\cite{grami},
which represents the state of the art for centralized FSM. 
GRAMI keeps state on a per-pattern basis, so few relatively straightforward 
changes to the code-base were sufficient to derive
a TLP implementation where patterns are partitioned
across a set of distributed workers.

GRAMI uses  a number of optimizations that are specific to FSM. 
In particular, it avoids materializing
all embeddings related to a pattern, a common approach for TLP algorithms.
Whenever a new pattern is generated, its instances are re-calculated on the fly, stopping
as soon as a sufficient number of embeddings  to pass the frequency threshold is found.
GRAMI thus solves a simpler problem than the TLV and \name\ implementations 
of FSM: it does not output all frequent embeddings but only their patterns.

The TLP version of GRAMI is significantly faster than TLV: GRAMI runs in 3 seconds for the same 
input graph and support compared to the hundreds of seconds for TLV. 
However, TLP suffers from extremely limited scalability, thus the performance can't be improved compared to the centralized
algorithm as seen in Figure~\ref{fig:scalability_tlv}.
This is again because of hot-spots: it is quite common
that only a few frequent patterns exist. Thus, irrespective of the size
of the cluster, only a few workers (equal to the number of
these frequent patterns) will be used. 
In addition, due to the skewed nature of many graphs, load is unlikely to be well balanced among these patterns.
Some patterns will typically be much more popular than others and the
corresponding workers will be overloaded.
Spreading the load by using techniques like work stealing is not viable since there exists no
straightforward way to split the work associated with a pattern in GRAMI.
The same problem holds for the other example applications we consider.
For instance, in the case of Motifs, for depth equal to 3, there are only $2$ patterns to process. 
While TLP can provide the best performance for a single thread (centralized) scenario, 
its lack of scalability limits the usefulness in
distributed frameworks. 

\subsection{\name: The TLE Paradigm}

We now focus on evaluating the performance of \name\ and its optimizations.
While Arabesque is generic enough to describe easily most graph mining
algorithms, the internals hide a powerful optimized engine with a
number of innovations that allows the system to efficiently process the
trillions of embeddings that graph exploration generates. 

\spara{Single Thread Execution Performance:} 
Arabesque is built from scratch
as a generic distributed graph mining system.
Since it has been observed that a centralized
implementation can outperform distributed frameworks that utilize
thousands of threads~\cite{scalability_cost}, we next show that the performance of Arabesque 
running on a single thread is
comparable to the best available centralized implementations.
 For Motifs we use G-Tries~\cite{gtries} as the centralized implementation. For Cliques 
 we use Mace~\cite{mace}. For FSM we use GRAMI~\cite{grami} to discover the frequent patterns and 
then VFLib~\cite{Cordella2004} to discover the embeddings. The centralized implementations
are highly optimized C/C++ implementations with the exception of GRAMI which uses Java.  
 
 For these experiments, we run Arabesque on a single worker with a single
 thread. We report the total time 
excluding the start-up and shutdown overhead to run a Giraph job in Hadoop (on average 10 seconds).

\begin{table*}[!ht]
\centering
\begin{tabular}{@{}llr@{}}
\toprule
  &  \textbf{Centralized}			& \textbf{\name} \\
 \textbf{Application} 	& \textbf{Baseline} & \textbf{(1 thread)}\\
\toprule
Motifs (MS=3) & G-Tries: 50s & 37s \\
Cliques (MS=4) & Mace: 281s & 385s \\
FSM (S=300) & Grami+VFLib: 3s\,+\,1.8s & 5s \\
FSM (S=220,MS=7) & Grami+VFLib: 13s\,+\,1,800s & 8,548s \\
 \bottomrule
\end{tabular}
\caption{Execution times of centralized implementations and \name\ running on a single thread. 
Motifs and Cliques were run with the MiCo graph, FSM with CiteSeer.}
\label{tab:single-thread}
\end{table*}

Table~\ref{tab:single-thread} shows a comparison between 
baseline single-threaded implementations and \name.
Even when running on a single thread, \name\ has comparable performance or is even faster than 
most centralized implementations. 
The only exception is GRAMI, which, as discussed, uses 
extra application-specific TLP optimizations 
and solves a simpler problem by not outputting all frequent embeddings. The performance advantage 
of GRAMI disappears when we require discovery of the actual embeddings as we see from the running time
of VFLib. 

These results are a clear indicator of the efficiency 
of \name. Despite being built with a
generic framework running over Hadoop (and Java), \name\ can achieve
performance comparable, and sometimes even superior, to application-specific implementations.
The main contributing factor, as we show later in this section, is that the user-defined 
functions in \name\ consume an insignificant amount of CPU time. 
The user-defined functions steer the exploration process through a high-level API.
The API abstracts away the details of the actual exploration process,
which are under the control of \name\ and can thus be efficiently implemented. 
This is in stark contrast to the graph processing systems analyzed in~\cite{scalability_cost}, where 
the user-defined functions perform
the bulk of the computation, leaving little room for system-level optimizations. 

\begin{table*}[!ht]
\centering
\iftoggle{singlecolumn}{%
  \footnotesize
}{%
}
\begin{tabular}{@{}llrrrrr@{}}
\toprule
\iftoggle{singlecolumn}{%
		\textbf{Application - Graph} & \textbf{Centralized Baseline} & \multicolumn{5}{c}{\textbf{Arabesque - Num. Servers (32 threads)}}\\
}{%
		\multirow{2}{*}{\textbf{Application - Graph}} & \multirow{2}{6em}{\textbf{Centralized Baseline}} & \multicolumn{5}{c}{\textbf{Arabesque - Num. Servers (32 threads)}}\\
}
		\cmidrule{3-7}  
& &  \textbf{1} 	& \textbf{5} 	& \textbf{10} 	& \textbf{15} 	& \textbf{20}\\
\toprule
Motifs - MiCo 	& G-Tries: 8,664s & 328s	& 74s	& 41s	& 31s	& 25s\\
FSM - CiteSeer	& Grami+VFLib: 13s + 1,800s & 431s	& 105s   & 65s	& 52s 	& 41s\\
Cliques - MiCo	& Mace: 14,901s & 1,185s	& 272s	& 140s	& 91s	& 70s\\
Motifs - Youtube & G-Tries: {\em Fail} & 8,995s & 2,218s & 1,167s & 900s  & 709s \\
FSM - Patents	& Grami+VFLib: 1,147s + $>$19h & 548s & 186s   & 132s	& 102s 	& 88s\\
 \bottomrule
\end{tabular}
\caption{Scalability of Arabesque - For FSM - CiteSeer, 
the chosen support was 220 and the search was terminated at embedding size 7, while for FSM - Patents the chosen 
support was 24k with no maximum embedding size.}
\label{tab:scalability}
\end{table*}

\begin{figure}[!ht]
\centering
\begin{tikzpicture}
    \pgfplotstableread[col sep=comma]{speedup5-motifs-ezip-mico.data}\speedupmotifs
    \pgfplotstableread[col sep=comma]{speedup5-gspan-ezip-citeseer.data}\speedupgspan
    \pgfplotstableread[col sep=comma]{speedup5-cliques-ezip-mico.data}\speedupcliques
    \pgfplotstableread[col sep=comma]{speedup5-motifs-ezip-youtube.data}\speedupmotifsyt
    \pgfplotstableread[col sep=comma]{speedup5-gspan-ezip-patents.data}\speedupgspanpatents
    \pgfplotstableread[col sep=comma]{
	nodes, speedup
	5, 1
	10, 2
	15, 3
	20, 4
    }\speedupideal
    \begin{axis}[%
	width=\linewidth,
	height=1.65in,
	legend style={at={(0.5,1.05)}, anchor=south, legend columns=3, font=\scriptsize},
	bar width=1.5em,
	xmin=5,
	ymin=1,
	ymax=4,
	xmax=20,
	ylabel=Speedup,
	xlabel=Number of nodes (32 threads),
	x label style={font=\scriptsize},
	y label style={font=\scriptsize},
	y tick label style={align=center,font=\scriptsize},
	x tick label style={align=center,font=\scriptsize,text width=4em},
	every node near coord/.append style={fill=white, fill opacity=0.7, anchor=south,font=\tiny, yshift=-0.15em},
	enlarge x limits=0.075,
	xtick=data,
	mark options={solid}
	]
    \addplot[color=black, mark=*]
	table[x=nodes, y=speedup]{\speedupideal}; 
    \addplot[color=pastelblue, mark=diamond*]
	table[x=nodes, y=speedup]{\speedupmotifs}; 
    \addplot[color=pastelred, mark=triangle*]
	table[x=nodes, y=speedup]{\speedupgspan}; 
    \addplot[color=pastelgreen, mark=square*]
	table[x=nodes, y=speedup]{\speedupcliques}; 
    \addplot[color=pastelblue, mark=diamond*, dashed]
	table[x=nodes, y=speedup]{\speedupmotifsyt}; 
    \addplot[color=pastelred, mark=triangle*, dashed]
	table[x=nodes, y=speedup]{\speedupgspanpatents}; 
    \legend{Ideal, Motifs (MiCo), FSM (CiteSeer), Cliques (MiCo), 
	    Motifs (Youtube), FSM (Patents), Cliques (Youtube),
	    }
    \end{axis}
\end{tikzpicture}
    \caption{Scalability of Arabesque:Speedup relative to the configuration with 5 servers.}
    \label{fig:scalability-speedup}
\end{figure}

\spara{Scalability:}
The TLE approach of Arabesque makes it easy to scale the system to a large number of servers.
We ran all three algorithms on datasets that allow computation to 
terminate in a feasible time on a single server while leaving sufficient 
work to be executed on a large cluster.
Table~\ref{tab:scalability} reports execution times, thus excluding setup and shutdown times, with a growing number of servers and, for reference, the  execution time of the 
centralized baselines.
Figure~\ref{fig:scalability-speedup} illustrates the same results in terms of speedup, comparing
distributed settings among each other.

The results show that \name\ scales to a large number of servers.
Different applications show different scalability factors.
In general, applications generating more intermediate state and more patterns
scale less.
For example, FSM scales less because it generates many patterns and 
transmits a large number of embeddings that are discarded
by the aggregation filter at the beginning of the next step, when aggregated metrics become available.
By contrast, in Cliques we have a single pattern at each step (a clique) and fewer embeddings.
The behavior of Motifs is in between. This trend is due to the characteristics of \ezips.
\name\ constructs one \ezip\ per pattern, and thus as the number of
patterns grows, so does the number of \ezips. Considering the same number of 
embeddings, the more \ezips\ they are split into, the smaller the potential for compression.
Furthermore, since \ezips\ are broadcast, the communication cost of transmitting embeddings
increases as more servers are added,
and the per-server computational cost of de-serializing and filtering out embeddings remains constant.

Despite these scalability limitations, \ezips\ typically remain advantageous.
We have tested the scalability of the system without \ezips\ and the slope of the speedup 
is closer to the ideal speedup than in Figure~\ref{fig:scalability-speedup}. 
Nevertheless, this better scalability is greatly outweighed by a significant increase
in the overall execution time, as we will see shortly.


\spara{\EZips:}
\name\ introduced \ezips\
to compress embeddings and make it possible to 
mine large graphs that generate trillions of embeddings (see Section~\ref{subsec:lossy-trees}).
Figure~\ref{fig:space-optim} shows the efficacy of \ezips\ by comparing the space required to store intermediate embeddings with and without \ezips\ 
at different exploration steps.
We report the sizes of the structures listing all the embeddings at the end of each superstep. 
This represents the minimum space required for the embedding: after de-serialization one can expect 
this value to grow much larger.
The results clearly show that \ezips\ can reduce memory cost by several orders of magnitude even in relatively 
small graphs such as CiteSeer. 



\begin{figure}[!ht]
\centering
\begin{tikzpicture}
    \pgfplotstableread[col sep=comma]{size-serialized-ezip-machine.data}\sizeezipciteseermachine
    \pgfplotstableread[col sep=comma]{size-serialized-noezip-total.data}\sizenoezipciteseertotal
    \pgfplotstableread[col sep=comma]{size-serialized-ezip-youtube-machine.data}\sizeezipyoutubemachine
    \pgfplotstableread[col sep=comma]{size-serialized-noezip-youtube-total.data}\sizenoezipyoutubetotal
    \begin{semilogyaxis}[%
	width=\linewidth,
	height=1.65in,
	legend style={at={(0.5,1.05)}, anchor=south, legend columns=2, font=\scriptsize},
	bar width=1.5em,
	xmin=1,
	ymax=100000,
	xmax=6,
	ytick={0.01, 0.1, 1, 10, 100, 1000, 10000, 100000},
	yticklabels={0.01, , 1, , 100, , 10000, },
	ylabel=Serialized size of embeddings (MB),
	xlabel=Exploration depth,
	log ticks with fixed point,
	x label style={font=\scriptsize},
	y label style={font=\scriptsize},
	y tick label style={align=center,font=\scriptsize},
	x tick label style={align=center,font=\scriptsize,text width=4em},
	every node near coord/.append style={fill=white, fill opacity=0.7, anchor=south,font=\tiny, yshift=-0.15em},
	enlarge x limits=0.075,
	xtick=data
	]
    \addplot[color=pastelblue, mark=*]
	table[x=depth, y=mb]{\sizeezipciteseermachine}; 
    \addplot[color=pastelred, mark=square*]
	table[x=depth, y=mb]{\sizenoezipciteseertotal}; 
    \addplot[color=pastelgreen, mark=triangle*]
	table[x=depth, y=mb]{\sizeezipyoutubemachine}; 
    \addplot[color=black, mark=diamond*]
	table[x=depth, y=mb]{\sizenoezipyoutubetotal}; 
    \legend{\Ezips\ (CiteSeer), No \ezips\ (CiteSeer), \Ezips\ (Youtube), No \ezips\ (Youtube)}
    \end{semilogyaxis}
\end{tikzpicture}
    \caption{Compression effect of \ezips\ processed at each depth. 
    Data captured from executions of FSM on
CiteSeer (S=220, MS=7) and Youtube (S=250k).  Includes a minor fix of the corresponding plot in the short version of this paper.}
	\label{fig:space-optim}
\end{figure}

As with any compression technique, \ezips\ trade space for computational costs, 
so one might wonder whether using \ezips\ results in longer execution times.
Indeed, the opposite holds: \ezips\ significantly speed up computation, 
especially when running many exploration steps.
Figure~\ref{fig:remove-optim} reports the normalized execution time slowdown when \ezips\ are disabled, compared to the results of Table~\ref{tab:scalability}. Removing \ezips\ can increase execution time 
up to $4$ times in these experiments. 
A more compact representation of the embeddings, in fact, results in less network overhead
to transmit embeddings across servers, lower serialization costs, and less overhead 
due to garbage collection.
We have observed, however, that in the first exploration steps
with very large and sparse graphs, the overhead of constructing ODAGs
outweights the cost of sending individual 
embeddings, because ODAGs achieve very little compression. 
In such cases, we can revert to using embedding lists.

\begin{figure}[!ht]
\centering
\begin{tikzpicture}
    \pgfplotstableread[col sep=comma]{optimizations-noezip.data}\optimizationsnoezip
    \begin{axis}[%
	width=\linewidth,
	height=1.65in,
	ybar,
	ymax=5,
	ymin=0,
	bar width=1.5em,
	ylabel=Relative slowdown factor,
	x label style={font=\scriptsize},
	y label style={font=\scriptsize},
	symbolic x coords={Motifs MiCo, FSM CiteSeer, Cliques MiCo, Motifs Youtube, FSM Patents},
	y tick label style={align=center,font=\scriptsize},
	x tick label style={anchor=north,align=center,font=\scriptsize,text width=3.5em},
	nodes near coords,
	every node near coord/.append style={anchor=south,font=\scriptsize, yshift=-0.15em},
	enlarge x limits=0.2,
	xtick=data]
    \addplot[fill=pastelblue, area legend]
	table[x=exec, y=normalizedtime]{\optimizationsnoezip}; 
    \end{axis}
\end{tikzpicture}
\caption{Slowdown factor when storing full embedding lists compared to the results of Table~\ref{tab:scalability} with 20 servers.}
\label{fig:remove-optim}
\end{figure}

\begin{figure}[!ht]
\centering
\begin{tikzpicture}
    \pgfplotstableread[col sep=comma]{optimizations-notwolevel.data}\optimizationsnotwolevel
    \begin{axis}[%
	width=\linewidth,
	height=1.65in,
	legend style={anchor=south, at={(0.5, 1.05)}, legend columns=-1, font=\scriptsize, /tikz/every even column/.append style={column sep=0.1cm}},
	ybar,
	ymax=52,
	ymin=0,
	bar width=1.5em,
	ylabel=Relative slowdown factor,
	x label style={font=\scriptsize},
	y label style={font=\scriptsize},
	symbolic x coords={{Motifs MiCo (MS=3)}, {Motifs Patents (MS=3)}, {FSM CiteSeer (S=220 MS=6)}, {FSM Patents (S=30k)}},
	y tick label style={align=center,font=\scriptsize},
	x tick label style={align=center,font=\scriptsize,text width=3.5em},
	nodes near coords,
	every node near coord/.append style={anchor=south,font=\scriptsize, yshift=-0.15em},
	enlarge x limits=0.2,
	xtick=data]
    \addplot[fill=pastelgreen, area legend]
	table[x=exec, y=normalizedtime]{\optimizationsnotwolevel}; 
    \end{axis}
\end{tikzpicture}
\caption{Slowdown factor when removing two-level pattern aggregation (not applicable to Cliques).}
\label{fig:remove-optim2}
\end{figure}

\begin{table*}[!ht]
\centering
\iftoggle{singlecolumn}{%
  \iftoggle{twoleveltiny}{%
  \tiny%
  }{
  \scriptsize%
  }%
}{%
  \small%
}
\iftoggle{singlecolumn}{%
\begin{tabular}{@{}p{3.5em}rrrrrr@{}}
}{%
\begin{tabular}{@{}lrrrrrr@{}}
}
\toprule
& \textbf{Motifs-MiCo} & \textbf{FSM-CiteSeer} & \textbf{FSM-CiteSeer}& \textbf{Motifs-MiCo}& \textbf{FSM-Patents}& 
\textbf{Motifs-Youtube}\\
& \textbf{MS=3} & \textbf{S=300} & \textbf{S=220,MS=7} & \textbf{MS=4} & \textbf{S=24k} & \textbf{MS=4}\\

\toprule
Embeddings 
\iftoggle{twolevelapprox}{%
& $6.6 * 10^7$ & $2.9 * 10^6$ & $1.7 * 10^9$ & $1.1 * 10^{10}$ & $1.9 * 10^9$ & $2.2 * 10^{11}$ \\
}{%
& 66,081,419 & 2,890,024 & 1,680,983,703 & 10,957,439,024 & 1,910,611,704 & 218,909,854,429 \\
}
\iftoggle{singlecolumn}{%
\midrule
}{}
Quick patterns 
& 3 & 116 & 1,433 & 21 & 1,800 & 21 \\
\iftoggle{singlecolumn}{%
\midrule
}{}
Canonical patterns
& 2 & 28 & 97 & 6 & 1,348 & 6 \\
\iftoggle{singlecolumn}{%
\midrule
}{}
\emph{Reduction factor}
\iftoggle{twolevelapprox}{%
& \emph{$2.2 * 10^7$}x & \emph{$2.5 * 10^4$}x & \emph{$1.1*10^6$}x & \emph{$5.2 * 10^8$}x & \emph{$1.1 * 10^6$}x &
\emph{$1*10^{10}$}x \\
}{%
& \emph{22,027,140}x & \emph{24,914}x & \emph{1,173,052}x & \emph{521,782,810}x & \emph{1,061,451}x &
\emph{10,424,278,782}x \\
}
\bottomrule
\end{tabular}
\caption{Effect of two-level pattern aggregation. 
The results refer to the deepest exploration level.}
\label{tab:2-level}
\end{table*}

\spara{Two-Level Pattern Aggregation:}
In Section~\ref{sec:optim}, we introduced our novel two-level
pattern-key aggregation technique to reduce the number of pattern canonicality checks,
i.e., graph isomorphism, run by the system.
Table~\ref{tab:2-level} compares the number of checks without the optimization,
which is equal to the number of embeddings, and with the optimization, which
is equal to the number of quick patterns.
The results show a reduction of several orders of magnitude using 
the optimization. For instance, for Motifs with the Youtube graph the optimization allows \name\
to run graph isomorphism only 21 times instead of $218$ billion times.
The number of quick patterns is also very close to the number of actual canonical patterns, thus minimizing 
the required number of graph isomorphism checks. 

The actual savings in terms of execution time depend not only on how often we compute graph isomorphism
but also on the cost of the computation itself, which in turn depends on the complexity of 
the pattern. 
In order to take this into account,
Figure~\ref{fig:remove-optim2} reports the relative execution time slowdown when two-level aggregation is disabled.
We consider smaller instances than in Table~\ref{tab:scalability} to keep execution times manageable;
the slowdown grows with larger instances.
The results show that without the optimization the system can be more than one order of magnitude slower,
since it spends most of its CPU cycles on computing graph isomorphism.

\begin{figure*}[!ht]
    \centering
    \subfloat[][{\raggedright FSM CiteSeer (S=220,MS=7)}]{
	\centering
	\begin{tikzpicture}
	    \pie[radius=1.3, color={pastelred, pastelgreen, pastelblue, pastelpurple, pastelgrey}]{5/R, 8/G, 11/C, 26/P, 50/W}
	\end{tikzpicture}
    }
    \qquad
    \subfloat[][Motifs MiCo (MS=4)]{%
	\centering
	\begin{tikzpicture}
	    \pie[radius=1.3, color={pastelred, pastelgreen, pastelblue, pastelpurple, pastelgrey}]{1/R (1\%), 18/G, 18/C, 15/P, 48/W}
	\end{tikzpicture}
    }
    \qquad
    \subfloat[][Cliques MiCo (MS=5)]{%
	\centering
	\begin{tikzpicture}
	    \pie[radius=1.3, color={pastelred, pastelgreen, pastelblue, pastelgrey}]{4/R, 59/G, 12/C, 25/W}
	\end{tikzpicture}
    }
\caption{CPU utilization breakdown during the superstep preceding the last one. W = Writing embeddings (\ezip~creation, serialization, transfer);
R = Reading embeddings (\ezip~extraction); G = Generating new candidates; C = Embedding
canonicality checking; P = Pattern aggregation.}
\label{fig:breakdownG1}
\end{figure*}

\spara{Execution time breakdown:}
The CPU utilization breakdown of 
Figure~\ref{fig:breakdownG1} shows that storing, sharing, and extracting embeddings occupies
a predominant fraction of CPU utilization.
Embedding canonicality and pattern canonicality checking still take a significant fraction of CPU cycles
even after our optimizations, showing that executing these checks efficiently is critical. 
Note that Cliques does not use pattern aggregation.
Interestingly, the user-defined functions consume an insignificant amount of CPU,
although their logic is fundamental in determining the exploration process and thus the overall system load.

\subsection{Large Graphs with \name} 
\label{sub:large_graphs}

We complete our evaluation by running \name\ on large graphs and checking 
the limits in terms of required resources. 
We use the SN and the Instagram graph for this evaluation. SN is both a large and dense
graph with an average degree of 79, while Instagram has close to one billion edges but is significantly 
less dense compared to SN. For these two graphs
we don't have real world labels, so we focus on graph mining problems that look for structural patterns,
such as Motifs and Cliques, rather than more inherently label-dependent problems such as FSM.


\begin{table}[!t]
\centering
\begin{small}
\begin{tabular}{@{}lrrr@{}}

\toprule
\textbf{Application} & \textbf{Time} & \textbf{Memory} & \textbf{Embeddings} \\
\toprule
Motifs-SN (MS=4) & 6h 18m & 110 GB & $8.4 *10^{12}$ \\
Cliques-SN (MS=5) & 29m & 50 GB & $3 * 10^{10}$ \\
Motifs-Inst (MS=3) & 10h 45m & 140 GB&  $5*10^{12}$ \\
\bottomrule

\end{tabular}
\end{small}
\caption{Execution details with large graphs and 20 servers.}
\label{tab:big-graphs}
\end{table}

In Table~\ref{tab:big-graphs}, we report the running time, the maximum memory used and the number of interesting embeddings 
that \name\ processed. 
For Motifs-SN, the application analyzed $8.4$ trillion embeddings and ran for 6 hours 18 minutes. Cliques, as expected, posed a smaller load on the system, and it ran in half an hour, analyzing 30 billion embeddings.
Instagram is a large and sparse graph, so \ezips\ do not have high compression 
efficiency in the first exploration steps.
In fact, we could not run Motifs with \ezips\ and MS = 4 because it exceeds 
the memory resources of our servers (256 GB). 
Table~\ref{tab:big-graphs} thus reports the results for MS = 3 using regular embedding lists.

Overall, the results show that even with the commodity servers that we utilize, \name\ can process 
graphs that are dense and have hundreds of millions of edges and tens of million of vertices. 




\section{Related Work}

Over the last decades graph mining has emerged as an important research
topic. Here
we discuss the state-of-the-art for the graph mining problems tackled in
this paper.
  
\spara{Centralized Algorithms:} 
Among the most efficient methods for frequent subgraph mining is 
gSpan~\cite{gspan}. However, gSpan is designed for mining multiple input graphs
and not a single large graph.
When there are multiple graphs, the frequency of a pattern is
simply the number of input graphs that contain it, so finding only one
instance of a pattern in a graph is sufficient to make this determination. 
If we instead have a single input graph, we have to find multiple instances
in the same graph, and this makes the problem more complex.
One of the first algorithms to mine patterns from a single graph was
proposed in \cite{kuramochi2005finding}. It uses a level-wise edge
growth extension strategy, but uses an expensive anti-monotonic
definition of support based on the maximal independent set to find
edge-disjoint embeddings.
For the single large input graph setting GRAMI~\cite{grami} is
a recent approach that is very effective. The motif problem was
introduced in~\cite{graphlets}. The work in \cite{gtries} proposes an
effective approach for storing and finding motif frequencies.
Listing all the maximal cliques is a well studied problem, with the
Bron-Kerbosch algorithm~\cite{bron-kerbosch} among the most efficient
ones in practice. See~\cite{eppstein2011listing} for a recent method
that can handle large sparse real-world graphs.

\spara{Distributed and Parallel Approaches:} 
Recently there have been
several papers on both parallel and distributed FSM using MPI or the
MapReduce
framework~\cite{wang2004parallel,di2006dynamic,lu2013efficiently,hill2012iterative,lin2014large,bhuiyan2015iterative}
as well as GPUs~\cite{2014-bigmine}.
However, all these methods focus on the case of multiple input
graphs, which is simpler as we have previously discussed.
Some existing work targets graph matching, a subset of graph mining 
problems: given a query $q$, it finds all its
embeddings in a distributed manner.
The work in~\cite{shao2014parallel} uses a Pregel-based approach for graph matching
in a single graph, while~\cite{zhao2012sahad} proposes a Hadoop-based
solution. 
 For motifs, \cite{aparicio2014parallel} proposes a multicore parallel
approach, while \cite{slota2014complex} develops methods for
approximate motif counting on a tightly coupled HPC
system using MPI. An early work on parallel maximal clique enumeration is~\cite{dahlhaus88}, 
which proposed a parallel CREW-PRAM
implementation. A more recent parallel algorithm on the Cray XT4
machine was proposed in \cite{schmidt2009scalable}. The work in
\cite{cheng2012fast} uses an MPI-based approach, whereas  
MapReduce based implementations are given in
\cite{wu2009distributed,ashraf_2013}. The work in \cite{bahmani2012densest} focuses
on the related problem of finding dense subgraphs using MapReduce. 

\section{Conclusions}
In this paper, we showed that distributing graph mining tasks is far from trivial.
Focusing on optimizing centralized algorithms
and then considering how to convert them to distributed solutions using a TLV or TLP approach
can result to scalability issues. 
Distributing these tasks requires a mental shift on how to approach these
problems. 

\name\ represents a novel approach to graph mining problems. It is a system designed 
from scratch as a distributed graph mining framework.
\name\ focuses simultaneously on scalability and on providing a user-friendly simple programming API
that allows non-experts to build graph mining workloads. This follows the spirit of
the MapReduce and Pregel frameworks that democratized the processing and analysis of big data.
We demonstrated that \name's simple programming API can be used to build highly efficient distributed 
graph mining solutions
that scale and perform very well. 

\section*{Acknowledgments}
We would like to thank Landon Cox for shepherding the paper, the SOSP reviewers 
for their valuable feedback, Ehab Abdelhamid for the GRAMI implementation, Nilothpal Talukder 
for some of the datasets, Tommi Junttila for changing the license of his jbliss library, 
 Lori Paniak and Kareem El Gebaly for the cluster setup.
Carlos  H. C. Teixeira would like to thank CNPq, Fapemig and Inweb for the travel support to attend the conference.

\bibliographystyle{acm}
\bibliography{references}

\newpage


\appendix
\section*{Appendix}
This section contains the proofs of the claims of Section~\ref{sec:system}.
Formally, our definition of canonicality is the following. 
We focus on the case of vertex-based exploration since the edge-based case is analogous.

\begin{definition} {\bf Canonical embedding.} Let $e = \langle v_1, \ldots, v_n \rangle$ be a canonical embedding $e$ induced by vertices.
We say that $e$ is canonical if the following three properties hold:
\begin{enumerate}
\item[P1:] $\forall i > 1$ it holds $v_1 < v_i$
\item[P2:] $\forall i > 1$ it holds $\exists j < i:(v_j,v_i) \in E$
\item[P3:] $\forall v_h, v_k, v_j$ such that $(v_h, v_j) \in E$ and $h < k < j$ and $\not\exists w < h : (v_w, v_j) \in E$, it holds $v_k < v_ j$.
\end{enumerate}
\label{def:canonical-e2}
\end{definition}

We now first show that Algorithm~\ref{alg:check-canonical} verifies Definition~\ref{def:canonical-e2} and then that Definition~\ref{def:canonical-e2} satisfies Uniqueness and Extendibility.

\begin{theorem}
Algorithm~\ref{alg:check-canonical} returns true if and only if Definition~\ref{def:canonical-e2} holds.
\label{thm:algo}
\end{theorem}

\begin{proof}
The proof is by induction on the length of the embedding $e$.
If $e$ consists of only one vertex, it is trivial to observe that Algorithm~\ref{alg:check-canonical} implements Property P1 of Definition~\ref{def:canonical-e2} by checking if $v_1 > v$.
For the inductive step, let $e = \langle v_1 \ldots v_n, v \rangle$ with $n \geq 1$.
By construction, Algorithm~\ref{alg:check-canonical} has already returned true on $e' = \langle v_1 \ldots v_n \rangle$,
so we know that $e'$ satisfies Definition~\ref{def:canonical-e2} by inductive hypothesis.
In order to verify Definition~\ref{def:canonical-e2} on $e$ it is thus sufficient to check that the Properties P2 and P3 hold for $v_j = v$.
The algorithm scans $e$ until it finds a neighbor of $v$, which corresponds to the $v_h$ of Property P3, and then returns false if and only if some $v_k > v$ is found with $i = k > h$.
Property P2 holds because $v$ is an extension of the parent embedding by construction.
\end{proof}

\begin{theorem}
Definition~\ref{def:canonical-e2} satisfies Extendibility.
\end{theorem}
\begin{proof}
Follows directly from the incrementality of Algorithm~\ref{alg:check-canonical} and from Theorem~\ref{thm:algo}.
\end{proof}

\begin{theorem}
Definition~1 satisfies Uniqueness
\end{theorem}

\begin{proof}

We first show that there exists a canonical embedding for each embedding $e$.
Let $e_c$ be the canonical embedding automorphic to $e$. 
We can construct $e_c = \langle v_1^c, \ldots, v_n^c \rangle$ as follow: 
$v_1^c$ will be the vertex of $e$ with the smallest ID overall. Moreover, $v_k^c$ will be 
the vertex of $e$ with the smallest ID such that 
$v_k^c$ is a neighbor of some vertex in $e_c^{k-1} = \langle v_1^c, \ldots, v_{k-1}^c \rangle$ and $v_k^c \notin e_c^{k-1}$.
By construction, $e_c$ follows Definition~\ref{def:canonical-e2}. Now, we will prove that $e_c$ is unique.

We now show that there exist no two canonical embeddings for the embedding $e$.
The proof is by contradiction.
Suppose that there exist two canonical embedding for $e$, $a = \langle v_1^a, \ldots, v_n^a \rangle$ and $b = \langle v_1^b, \ldots, v_n^b \rangle$.
By Property P1 of Definition~\ref{def:canonical-e2}, both embeddings must start with the vertex with smallest ID, otherwise we are done.
Let $k > 1$ be the position of the first disagreement between $a$ and $b$ such that $v_k^a \neq v_k^b$. 
Assume w.l.o.g. that $v_k^a < v_k^b$.
By Property P2 of Definition~\ref{def:canonical-e2}, there exists a neighbor of $v_k^a$ in $a$ with index $h < k$.
Since $k$ is the first divergence between $a$ and $b$, it holds $v_h^a = v_h^b$.
In addition, since $a$ and $b$ have the same set of vertices, $v_k^a$ occurs in $b$ at a position $j > k$,
so $v_j^b = v_k^a$.
By Property P3 of Definition~\ref{def:canonical-e2} it must hold that $v_k^b < v_j^b$, since $v_h^b = v_h^a$ is a neighbor of $v_j^b = v_k^a$. 
However, it holds $v_k^b > v_k^a = v_j^b$, a contradiction.
\end{proof}

We now argue that \name's distributed exploration strategy 
satisfies the completeness property of embedding exploration (see
Section~\ref{sec:model}).
The argument relies on the anti-monotonic and
automorphism invariance properties of the user-defined functions.

\begin{theorem} For each embedding $e$ such that
    $\phi(e) \wedge \alpha(e) = \textit{true}$,
    it holds that Algorithm~\ref{alg:exploration} adds $\pi(e)$ and $\beta(e)$ to $O$.
    \end{theorem}

\begin{proof} We consider vertex-based exploration, but the edge-based
    exploration case is analogous.
    We know by uniqueness that
    there exists a canonical automorphism $e_c$ of $e$. By the automorphism invariance property of 
    filters,
    we have that the filter functions evaluate true for $e_c$ as they do for $e$. 
    We now show that \name\ obtains
    $e_c$ by exploring a sequence of canonical
    embeddings $e_1 \ldots e_n$ with $e_n = e_c$ such that $e_i$
    is a candidate embedding at exploration step $i$ and
    consists of $i$ vertices. At the first exploration step, $C$ contains all vertices,
    so a one-vertex embedding $e_1$ consisting of one of the vertices in $e_c$ is in $C$. 
    All filter functions evaluate true for $e_1$ because of anti-monotonicity, 
    since they evaluate true for $e_c = e_n$.
    In addition, all one-vertex embeddings are canonical.
    Algorithm~\ref{alg:exploration} thus adds $e_1$ to $F$ and expands it 
    in the following exploration step, where it
    obtains a new candidate set $C$ containing all extensions of
    $e_1$ with two vertices. Since all embeddings we consider are connected, 
    some $e \in C$ includes two vertices of $e_c$.
    By extendibility, the canonical automorphism $e_2$ of $e$ is also included in $C$.
    Again, all filter functions evaluate true for $e_2$ because of anti-monotonicity, 
    since they evaluate true for $e_c = e_n$.
    Because of this and of the canonicality of $e_2$, Algorithm~\ref{alg:exploration}
    adds $e_2$ to $F$ and expands it in the next step.
    The same argument can be
    repeated inductively to show that Algorithm~\ref{alg:exploration} 
    includes $e_n = e_c$ in $C$, adds $\pi(e_c)$ to $O$, and adds $e_c$ to $F$ at exploration step $n$,
    and also that it adds $\beta(e_c)$ to $O$ at step $n+1$. 
    By automorphism invariance, $\pi(e_c) = \pi(e)$ 
    and $\beta(e_c) = \beta(e)$. \end{proof}

\end{document}